\documentclass[11pt]{article}

\usepackage{fullpage}
\usepackage{amsmath}
\usepackage{amsthm,amsfonts,amssymb}
\usepackage{xspace}
\usepackage{bbm}
\usepackage{subcaption}
\usepackage{graphicx}
\usepackage{bbm}
\usepackage{hyperref}   
\usepackage{etoolbox}
\usepackage{enumitem}
\usepackage{authblk}

\usepackage[noend]{algorithmic}
\usepackage[ruled,vlined]{algorithm2e}
\usepackage{comment}
\usepackage[authoryear,square]{natbib}
\setcitestyle{square,comma}  

\usepackage{float}

	\hypersetup{
			colorlinks=true,
			linkcolor=[rgb]{0,0,.5},
			urlcolor=[rgb]{0,0,.5},
			citecolor=[rgb]{0,0,.5},
			pdfstartview=FitH}

\usepackage{accents}

\usepackage{graphicx} 
\usepackage{booktabs} 

\usepackage{caption}

\usepackage{xcolor}

\usepackage{pgfplots}

\pgfplotsset{compat=1.10}

\usepackage{tikz}

\usepackage{enumitem}
\usepackage{mathtools}
\usetikzlibrary{quotes,positioning}
\allowdisplaybreaks

\newtheorem{theorem}{Theorem}[section]
\newtheorem*{theorem*}{Theorem (Informal)}
\newtheorem{lemma}{Lemma}[section]
\newtheorem{corollary}{Corollary}[section]
\newtheorem{definition}{Definition}[section]
\newtheorem{remark}{Remark}[section]

\newtheorem{assumption}{Assumption}[section]

\newcommand{\supp}{\operatorname{supp}}
\newcommand{\R}{R^{(\mathcal{D}^i)}}
\newcommand{\RSwap}{R^{(\mathcal{D}^i)}_{\operatorname{Swap}}}

\newcounter{numrellocal}
\renewcommand{\thenumrellocal}{\arabic{numrellocal}}
\newcounter{numrelglobal}

\makeatletter
\newcommand{\numrel}[2]{
  \stepcounter{numrellocal}
  \refstepcounter{numrelglobal}
  \ltx@label{#2}
  \overset{(\thenumrellocal)}{#1}
}
\makeatother

\usepackage{xcolor}

\usepackage{thmtools} 
\usepackage{thm-restate}

\usepackage{amsmath}
\usepackage{pgfplots}

\title{Prior-Agnostic Incentive-Compatible Exploration}

\author{Ramya Ramalingam\thanks{Supported in part by the AWS AI ASSET Gift for Research in Trustworthy AI.}}
\author[1]{Osbert Bastani}
\author[1]{Aaron Roth}
\affil[1]{Department of Computer and Information Sciences, University of Pennsylvania}

\date{} %

\begin{document}
\maketitle

\begin{abstract}
In bandit settings, optimizing long-term regret metrics requires exploration, which corresponds to sometimes taking myopically sub-optimal actions. When a long-lived principal merely \emph{recommends} actions to be executed by a sequence of different agents (as in an online recommendation platform) this provides an incentive misalignment: exploration is ``worth it'' for the principal but not for the agents. Prior work studies regret minimization under the constraint of Bayesian Incentive-Compatibility in a static stochastic setting with a fixed and common prior shared amongst the agents and the algorithm designer. 

We show that (weighted) swap regret bounds on their own suffice to cause agents to faithfully follow forecasts in an approximate Bayes Nash equilibrium, even in dynamic environments in which agents have conflicting prior beliefs and the mechanism designer has no knowledge of any agents beliefs. To obtain these bounds, it is necessary to assume that the agents have some degree of uncertainty not just about the rewards, but about their arrival time --- i.e. their relative position in the sequence of agents served by the algorithm. We instantiate our abstract bounds with concrete algorithms for guaranteeing adaptive and weighted regret in bandit settings. 
\end{abstract}

\section{Introduction}
\label{sec:intro}
The central tension in learning with restricted feedback (i.e. ``bandit learning'') is managing \emph{exploration} and \emph{exploitation}. Informally, good learning algorithms must sometimes play actions that seem myopically sub-optimal in order to collect information that helps them guarantee approximate optimality in the long run. Exploration is therefore aligned with the goals of an algorithm designer who values long-term performance (measured by \emph{regret bounds}) despite its potential short-term sub-optimality. But this incentive alignment breaks when the actions of the learning algorithm are used as \emph{recommendations} for a sequence of agents who interact with the algorithm only briefly, as is the case for learning algorithms that are used to generate recommendations on online platforms. Although the algorithm designer still cares about long-term performance guarantees like regret, each agent cares myopically about the reward on the round at which they are served (since future reward does not accrue to them). Thus they may not be incentivized to follow the recommendations of the learning algorithm. This would be catastrophic to the learning algorithm since the data it uses to learn is obtained only through the actions of the agents it gives recommendations to. 

This incentive mismatch motivates the problem of Bayesian Incentive-Compatible Bandit Exploration \citep{wisdomofcrowd, MansourSlivkins_BIC}. In this model, actions are associated with stochastic rewards whose means are drawn from a prior distribution that is commonly known to  each of the agents and the algorithm designer. The goal is to design an algorithm that guarantees strong long-term regret bounds while simultaneously guaranteeing to each agent that following the recommendation that they are given is \emph{Bayesian Incentive Compatible} --- i.e. that it is their best response under their posterior beliefs, given that all previous agents also followed their recommended action. \citet{MansourSlivkins_BIC} show how to obtain strong regret bounds subject to the Bayesian Incentive Compatibility constraint in this model, subject to modest restrictions on the common prior distribution.

However, the information environment that \citep{wisdomofcrowd, MansourSlivkins_BIC} work in is restrictive. They assume that the environment is static (i.e. that it is a stochastic bandit instance), that all agents have a common prior over its parameters, and that this prior is shared by the algorithm designer who is able to use it in designing the algorithm. Even the ``detail-free'' variant of \cite{MansourSlivkins_BIC} requires some knowledge of the prior (the ordering of arms according to the prior mean and bounds on the mean rewards, and the assumption that the prior is independent across actions), and that this be consistent across agents. In this paper we aim to relax all of these assumptions, and to give guarantees for truly prior agnostic algorithms. In particular, we generalize the stochastic setting of prior work to allow for dynamically changing environments, we do not require that agents share the same prior, and the algorithms we analyze do not require \emph{any} knowledge of the prior of any of the agents (although we require that the agent priors satisfy (generalizations of) the same kinds of mild feasibility conditions as \cite{MansourSlivkins_BIC}). Since the guarantees we give will hold for any prior satisfying these feasibility conditions, we can thus also view the agents as being agnostic to the prior (believing only that it satisfies the feasibility conditions).

To obtain these generalizations, we introduce two relaxations. First, we study agents who do not have complete certainty over their exact position in the arrival sequence, but instead believe that their arrival position is drawn according to an agent-specific temporal prior distribution, which must be sufficiently dispersed. Second, our incentive compatibility guarantees will be that the policy of each agent faithfully following the recommended action will be an \emph{approximate} Bayes Nash Equilibrium. We view our assumption that agents have uncertainty about their arrival time as mild. We do  not require that agents are unaware of the ``wall clock time'' at the time they use the system --- what our model really requires is uncertainty about their \emph{relative position} in the sequence of agents who are served by the recommendation engine. Prior work gives guarantees (in more restrictive settings) even to agents who are certain that they are the 113,495'th customer to use the recommendation engine; we relax this to (roughly speaking\footnote{If the total number of agents is $T$, then our theorems can accommodate uncertainty that is (say) uniform over some sub-interval of length $\omega(\sqrt{T})$.}) customers who know their rough position but have uncertainty on the order of $\pm 500$. We are not the first to introduce temporal uncertainty on behalf of the agents; \citet{wisdomofcrowd} study a (simple two action, deterministically fixed reward) setting in which agents believe they are uniformly distributed over blocks of time. \citet{priceofic_slivkins} also discuss arrival-time uncertainty as a realistic assumption, though they do not make use of it. 

Our starting point is the observation that the incentive compatibility constraint and the standard guarantee of \emph{swap regret} \citep{foster1999regret,blum2007external} are identical, except that the swap regret guarantee is \emph{marginalized} over the sequence of rounds in which the algorithm makes each particular recommendation. Informally, incentive-compatibility requires that for each action $a$, conditional on an agent being recommended to play $a$, $a$ must be a best response in expectation over the posterior distribution. Swap regret requires that for each $a$, $a$ must be a best response on the empirical distribution on outcomes conditional on the algorithm recommending $a$. It follows that an algorithm that guarantees zero swap regret on every sequence (and hence in expectation over every prior) implies Bayesian Incentive Compatibility for every agent, regardless of their beliefs on action rewards, so long as they believe that their arrival time is uniformly random---because this aligns the marginal promise of the swap regret guarantee with the individual agent beliefs\footnote{\citet{wisdomofcrowd} make a similar observation that in their two-action setting, under uniform arrival beliefs, agent incentives align with a planner maximizing average utility. In stochastic settings (and in two action adversarial settings) swap regret and external regret coincide, but they are different in non-stochastic settings with $3$ or more actions.}. More generally, it is possible to define a \emph{weighted} notion of swap regret that weights each time step differently, rather than uniformly marginalizing over time. Correspondingly, any algorithm that guarantees no weighted swap regret, as weighted by the arrival-time belief distribution of an agent, promises Bayesian Incentive-Compatibility for that agent. 

The observation we have just sketched gives a way to promise incentive compatibility guarantees to many downstream agents, independently of their beliefs on rewards, so long as they have arrival time beliefs that are sufficiently diffuse to be compatible with weighted regret guarantees. But it seems to have traded detailed dependence on one prior distribution (the reward distribution) for detailed dependence on another (the arrival time distribution) in order to be able to promise weighted regret guarantees. The key is that there are algorithms that can simultaneously guarantee diminishing weighted regret guarantees with respect to many different weighting functions \citep{blum2007external,lee2022online,noarov2023high}. Such algorithms can promise simultaneous Bayesian Incentive Compatibility guarantees for agents who have arbitrarily different reward priors and arbitrarily different arrival time priors, so long as the algorithm guarantees no weighted regret with respect to each agent's arrival time prior. For example, there exist optimized algorithms for guaranteeing  \emph{adaptive regret} in the bandit setting (see e.g. \cite{luo_intervalreg}) at optimal rates. In our language, adaptive regret corresponds to simultaneous weighted regret guarantees under all arrival time priors that are uniform over any contiguous sub-sequence of length sufficiently large (nontrivial regret guarantees arise for subsequences of length $L = \tilde \omega(\sqrt{T})$). More generally, if a mechanism guarantees no regret with respect to some collection of weighting functions, it also guarantees no regret to any weighting function in the convex hull of the collection; thus algorithms that promise no adaptive swap regret suffice to give incentive compatibility guarantees with respect to all agents whose arrival time priors can be written as convex combinations of uniform distributions over sequences of at least some minimum length $L$. 

Of course, learning algorithms will not be able to promise \emph{zero} regret --- they generally promise average regret guarantees that tend to $0$ with the length of the interaction, and some dispersion parameter of the weighting function. Their guarantees \emph{conditional} on recommending an action $a$ (which is what incentive compatibility hinges on) therefore also depend on the frequency with which $a$ is recommended, with no guarantees at all corresponding to exceedingly infrequent conditioning events. This is what obligates us to place the same kinds of assumptions on the agents' (different) prior reward distributions as \citet{MansourSlivkins_BIC} do --- informally that Agents' reward priors must allow some non-negligible probability that each action might be optimal. Since we no longer require static instances (i.e. agents may believe that rewards change over time), this assumption now is applied marginally to their reward beliefs, as marginalized over their arrival time beliefs. This assumption, combined with a belief that reward means are only ``slowly moving'' across time suffices for agents to believe that the expected frequency of each action's recommendation must be bounded away from zero, which is what lets us control (their beliefs about) regret guarantees even conditional on the recommended action. We note that this slowly moving assumption is a strict generalization of prior work, which studied stochastic settings, in which agents believe that reward means are invariant across time. Finally, under this same slowly moving assumption, we observe that weighted swap regret guarantees follow (up to some approximation error) from correspondingly weighted external regret guarantees, which allow us to apply existing optimized results about bandit algorithms that can guarantee diminishing adaptive regret guarantees. Putting all of these pieces together gives us theorems guaranteeing that for \emph{any} algorithm satisfying adaptive regret guarantees, faithfully following the algorithm's recommendations is an approximate Bayes Nash Equilibrium for all agents who have prior beliefs on reward and arrival time satisfying several mild assumptions, without the need that these prior distributions be equal across agents, or known either to the other agents or to the mechanism designer. 

\section{Related Work}
The problem of Bayesian Incentive Compatible Exploration was introduced in the seminal works of \citep{wisdomofcrowd, MansourSlivkins_BIC} (see also \cite{che2018recommender}). \citet{MansourSlivkins_BIC} give (under an additional product distribution assumption on the prior) ``detail free'' algorithms that do not require the mechanism designer to have full knowledge of the (common) prior held by the agents, but only the ordering of the prior arm means and bounds on the arm rewards. \citet{priceofic_slivkins} show that Thompson Sampling (initialized on the prior distribution) becomes incentive compatible if given a sufficiently large ``warm start'', where the length of the required warm start depends on the instance. A large number of variants of the basic model have been studied. \citet{Mansour_BayesianGames} study a model in which multiple agents arrive at each round and the recommendation parameterizes a Bayesian game, with a focus on exploring all ``explorable arms'' \citet{immorlica2019bayesian} study a model in which different agents have different preferences, with the aim of learning the optimal personalized recommendation policy. \citet{bahar2015economic,bahar2019social} overlay the model with a network in which agents can observe each others actions according to the network structure. \citet{immorlica2018incentivizing} make behavioral assumptions that let them relax assumptions of full rationality and commitment power. The model has also been extended to contextual and combinatorial bandit settings \citep{sellke2023incentivizing,CombinatorialBandit_BIC} and reinforcement learning settings \citep{simchowitz2024exploration}. Our primary point of departure from this literature is that we do not assume that the underlying environment is stochastic, do not require agents to agree in their beliefs, and do not require any knowledge of agent beliefs on the part of the mechanism designer. 
Rather than viewing regret guarantees as an objective to obtain under separate incentive compatibility constraints, we show that regret guarantees induce incentive compatibility constraints (under appropriate agent uncertainty on arrival time). 

A line of work has studied when \emph{greedy algorithms} (which do not balance exploration with exploitation, but only exploit) can nevertheless obtain strong regret bounds \citep{bastani2021mostly,kannan2018smoothed,raghavan2023greedy}. This kind of algorithm minimizes the strategic tension with myopic agents but requires strong assumptions on the underlying learning problem. 

Incentivizing exploration is an information design problem that can be viewed as a multi-round generalization of Bayesian Persuasion \citep{kamenica2011bayesian}. Several recent papers have given methods to remove Bayesian assumptions from Bayesian Persuasion models using calibration, which is closely linked to swap regret \citep{camara2020mechanisms,collina2024efficient}. These papers operate in a different setting, and deal with a single long lived agent, and so avoid the tension between the goals of a long lived principal and myopic agents which are the core of the incentivizing exploration problem. Our work is also conceptually related to the problem of ``predicting for downstream regret'' which aims to make forecasts that induce regret bounds in downstream agents with different objectives when they best respond to them \citep{noarov2023high,kleinberg2023u,roth2024forecasting,hu2024predict,lu2025sample}. This literature works in a full information setting and so does not have to grapple with how incentives interact with partial observation.

\section{Background and Notation}
We study a sequential decision-making setting where a learning algorithm interacts with different agents over each of $T$ rounds. In each round $t \in [T]$:
\begin{enumerate}
    \item A new agent arrives and receives a recommendation $I_t \in A$ from the algorithm, where $A$ is a fixed set of $K$ actions/arms.
    \item The environment selects a reward vector $u_t \in [0,1]^K$ independently of $I_t$, where $u_{t,a}$ is the reward of action $a$ at time $t$.
    \item The agent chooses an action $a_t \in A$, possibly deviating from the recommendation.
    \item The algorithm observes the played action $a_t$ and the corresponding reward $u_t(a_t)$. The agent receives utility equal to $u_t(a_t)$.
\end{enumerate}
The environment can be viewed as an adaptive adversary, and need not follow any stochastic process. The algorithms we study will have regret guarantees even against such adversaries. We index agents by $i \in [T]$ to disambiguate from time, indexed by $t \in [T]$. Each incoming agent $i$ reasons about the environment and the algorithm's recommendation through her own subjective beliefs. Firstly, she believes that rewards are generated according to a model specified by her prior $\mathcal{P}^i$,  under which the reward for playing action $a$ on day $t$ is drawn independently (but not necessarily identically) each day $t$ from a distribution with mean $\mu_{t,a}$. That is, each agent maintains a belief over a sequence of mean reward vectors $\mu = (\mu_1, \mu_2, \cdots, \mu_T)$ that captures the expected reward for playing any action on any day. We adopt this model of beliefs to accommodate both stationary and dynamic reward environments within a single framework. For an agent with stationary beliefs, $\mu_t$ is the same across all $t \in [T]$, but we also allow for agents who believe that rewards can change over time. Note that agent beliefs need not be correct --- this allows us to give guarantees in non-stationary settings in which agents have different beliefs.

\begin{definition}[Reward Belief]
A \emph{reward belief} $\mathcal{P}$ is a distribution over sequences of reward mean vectors $(\mu_1, \mu_2, \ldots, \mu_T)$, where each $\mu_t = (\mu_{t,a})_{a \in A} \in [0,1]^K$ defines the expected rewards for all $K$ actions at time $t$. An instance $\mu \sim \mathcal{P}$ models the reward for playing action $a \in A$ on day $t \in [T]$ as being drawn from a distribution with mean $\mu_{t,a}$. 
\end{definition}

Agent $i$ is also uncertain about her arrival time. She is aware of the overall protocol --- that the algorithm makes recommendations to one agent per round over $T$ rounds --- but does not observe her true arrival time $t$ prior to making a decision. We denote by $\tau_i$ the random variable that represents agent $i$'s arrival-time. Since she receives $I_t$, but is not aware of what $t$ is, we denote by $I_{\tau_i}$ the random variable that represents agent $i$'s belief over her recommendation. Further, her belief on the distribution of $I_{\tau_i}$ depends on how the algorithm behaves in earlier rounds, which in turn depends on who arrived before her and how they acted. To make these  expectations well-defined, we allow agents to hold beliefs over the \textit{entire} arrival sequence: 

\begin{definition}[Full Arrival Belief, Temporal Belief]
A \emph{full arrival belief} $\mathcal{T}$ is a joint distribution over the full arrival-time profile $\tau = (\tau_1, \tau_2, \cdots, \tau_T)$, with support over permutations, where $\tau_i$ is the random variable denoting when agent $i$ arrives. For a specific agent $i \in [T]$, define their \emph{temporal belief} $\mathcal{D}^i$ as the marginal distribution of $\tau_i$ induced by their full arrival belief $\mathcal{T}^i$, i.e. $\mathcal{D}^i(t) = \Pr_{\tau \sim \mathcal{T}^i}(\tau_i = t)$. 
\end{definition}
As is implicit in our notation, we assume that agent arrival time beliefs are independent of their reward beliefs. 
Note that the true sequence of rewards need not be drawn from any distribution defining an agent's beliefs, and that we allow different agents to have mutually incompatible beliefs. Likewise, an agent's temporal belief may not reflect their true arrival time --- it is even possible that it lies outside the support of their belief $\mathcal{D}$. We view agent uncertainty about their arrival time as a mild assumption. We do not require that agents are unaware of the ``wall-clock'' time at the moment that they are served---the uncertainty we model is uncertainty about an agents' relative position in the sequence of agents served by the recommendation algorithm; that is, the assumption is that agents are not certain about the precise number of other agents that have used the recommendation engine before them.  In order to keep track of the true sequence of events in real time, which is separate from agent beliefs, we define a transcript that captures the sequence of recommendations made by an algorithm, as well as agent outcomes and realized rewards:

\begin{definition}[Transcript]
A \textit{transcript} $\Pi_T = \{(I_t, a_t, u_{t})\}_{t=1}^T$ records the sequence of recommendations made by an algorithm, as well as the actual actions taken at each time-step and realized reward vectors.
\end{definition}

Though we record the full reward vector at each day, note that this is a bandit setting and the algorithm observes only $u_t(a_t)$. Given an algorithm and sequence of play, we can talk about the algorithm's performance in terms of \textit{regret}. There are a few different kinds of regret that will be relevant in this paper:

\begin{definition}[Weighted external regret]
\label{def:realized-regret}
Fix a transcript $\Pi_T = \{(I_t, a_t, u_t)\}_{t=1}^T$ of recommendations made by algorithm $\mathcal{A}$. For any temporal prior $\mathcal{D}$, the \emph{$\mathcal{D}$-weighted realized (external) regret} is
\begin{equation*}
{\mathrm{Reg}}_{\mathcal{D}}(\Pi_T)
\;:=\;
\max_{a\in A}\;
\sum_{t=1}^T \mathcal D(t)\,\bigl(u_{t,a}-u_{t,I_t}\bigr).
\end{equation*}
\end{definition}

\begin{definition}[Weighted swap-regret]
\label{def:swap-regret}
Fix a transcript $\Pi_T = \{(I_t, a_t, u_t)\}_{t=1}^T$ of recommendations made by algorithm $\mathcal{A}$. For any temporal prior $\mathcal{D}$, the \emph{$\mathcal{D}$-weighted swap-regret} is:
\begin{equation*}
\mathrm{SReg}_{\mathcal{D}}(\Pi_T)
\;:=\;
\max_{\phi\in \Phi}\;
\sum_{t=1}^T \mathcal{D}(t)\,\bigl(u_{t,\phi(I_t)}-u_{t,I_t}\bigr).
\end{equation*}
where $\Phi$ is the set of all swap-functions $\phi: A \to A$.
\end{definition}

These two notions of regret are defined in terms of realized rewards, and are (weighted versions of) standard metrics in online learning, and have to do with the algorithm's objective performance. However, when an agent decides what action to take, their reasoning is based not on the true rewards but their beliefs over them --- more specifically, the distributions from which rewards are drawn. To capture this perspective, we also consider a belief-based notion of regret defined with respect to an agent's reward belief:  

\begin{definition}[Weighted pseudo-regret, swap-regret]
\label{def:pseudo-regret}
Fix a transcript $\Pi_T = \{(I_t, a_t, u_t)\}_{t=1}^T$ of recommendations made by algorithm $\mathcal{A}$. For any temporal prior $\mathcal{D}$, and any \textit{realized} sequence of mean vectors $\mu = (\mu_1, \mu_2, \cdots, \mu_T)$, the \emph{$\mathcal{D}$-weighted pseudo-regret} with respect to reward belief instance $\mu$ is
\begin{equation*}
\mathrm{PReg}_{\mathcal{D}}(\Pi_T, \mu)
\;:=\;
\max_{a\in A}\;
\sum_{t=1}^T \mathcal{D}(t)\,\bigl(\mu_{t,a}-\mu_{t,I_t}\bigr).
\end{equation*}
Similarly, the \emph{$\mathcal{D}$-weighted pseudo swap-regret} with respect to $\mu$ is
\begin{equation*}
\mathrm{PSReg}_{\mathcal{D}}(\Pi_T, \mu)
\;:=\;
\max_{\phi\in \Phi}\;
\sum_{t=1}^T \mathcal{D}(t)\,\bigl(\mu_{t,\phi(I_t)}-\mu_{t,I_t}\bigr).
\end{equation*}
\end{definition}

Pseudoregret is a standard metric of regret used in settings with stochastic rewards --- here, we use it as a way to model agent beliefs over \textit{what they think} regret looks like, under the assumption that rewards are drawn from distributions that align with their beliefs. 

Having established our framework, what is a reasonable incentive-compatibility goal to ask for? Prior work  \citep{wisdomofcrowd,MansourSlivkins_BIC} proposes the property of Bayesian Incentive-Compatibility: 

\begin{definition}[Bayesian Incentive Compatibility]
An algorithm is \textit{Bayesian incentive-compatible} (BIC) for a sequence of agents with beliefs $\{(\mathcal{P}^i, \mathcal{T}^i)\}_{i=1}^T$ where for all $i \in [T]$, $\mathcal{D}^i(i) = 1$ and $0$ everywhere else, if
\begin{equation*}
    \mathbb{E}[\mu_{\tau_i,a} | I_{\tau_i} = a, \mathcal{P}^{i}, \mathcal{D}^{i}, \mathcal{E}_{i-1}] \geq \mathbb{E}[\mu_{\tau_i,b} | I_{\tau_i} = a, \mathcal{P}^{i}, \mathcal{D}^i, \mathcal{E}_{i-1}]
\end{equation*}
for all pairs of actions $b \neq a$ (for which the conditioning event has non-zero probability), and all $i \in [T]$, where $\mathcal{E}_{i-1}$ is the event that all agents prior to agent $i$ agreed to take their recommended action. 
\end{definition}
That is, given their priors,  the recommendation itself, and the fact that all previous agents followed their recommendations, each agent believes that the recommended action maximizes their expected reward. The above definition can be thought of as representing implementation in iterative elimination of dominated strategies, as it states that playing the recommended action is a dominant strategy \emph{if all prior players} have done the same. This is sensible for a model in which player arrival times are known with certainty, which corresponds to a sequential extensive form game. When player arrival times are uncertain, we instead are modeling a simultaneous move game. The actions in this game are mappings from recommendations to actions, and our solution concept will be Bayes Nash equilibrium.

\begin{definition}[Recommendation Game]
Fix an algorithm $\mathcal{A}$, and a collection of agent beliefs $\{(\mathcal{P}^i, \mathcal{T}^i)\}_{i=1}^T$. The algorithm induces a Bayesian game among the $T$ agents of the following form: 

At each round $t$, the algorithm issues a (possibly random) output $I_t$ to the agent arriving that round (call them $i$). Agent $i$ with beliefs $(\mathcal{P}^i, \mathcal{T}^i)$ privately observes their recommendation $I_{\tau_i}$ and then chooses an action via a \textit{strategy} $\sigma_i$. The strategy space for agent $i$ is $S_i = A^A$, the set of all functions $\sigma_i: A \to A$ mapping recommendations to actions. Agent $i$'s expected payoff from a strategy profile $\sigma = (\sigma_1, \sigma_2, \cdots, \sigma_T)$ over all agents is:
\begin{equation*}
    U_i(\sigma) = \mathbb{E}[\mu_{\tau_i,\sigma_i(I_{\tau_i})} | \sigma_{-i}]
\end{equation*}
where conditioning on $\sigma_{-i}$ is shorthand for taking the expectation over the distribution of transcripts induced by agent $i$'s prior and the strategy profiles $\sigma_{-i}$ for the other agents, and their expected payoff conditioned upon observing recommendation $a \in A$ is:
\begin{equation*}
U_i(\sigma | a) = \mathbb{E}[\mu_{\tau_i, \sigma_i(I_{\tau_i})} | I_{\tau_i} = a, \sigma_{-i}].
\end{equation*}
\end{definition}

\begin{remark}
The conditional expectation defined above may in general be very complicated to analyze. Forming a posterior about $\mu_{\tau_i}$ requires agent $i$ to reason not only about the reward distributions at their unknown arrival time, but also the reward distributions and actions of the agents who arrived at all rounds that preceded them --- because the algorithm's recommendation distribution is a function of the prefix of the transcript that it has seen so far. For general strategy profiles $\sigma$, which actions are taken (as a function of recommendations made at prior rounds) depend on the identity of who arrives before, and when prior agents do not follow the recommendations of the algorithm, this may disrupt regret guarantees the algorithm might have otherwise had. 

However these expectations become much easier to analyze when agents follow the \textit{compliant} strategy profile $\sigma^* = (\sigma_1^*, \sigma_2^*, \cdots, \sigma_T^*)$, where $\sigma_i^*$ is the identity function for all $i \in [T]$. When agents do this, the distribution on recommendations that the algorithm makes at day $t$ no longer depends on the identity or order of the agents that arrived previously, because they all play the same compliant strategy. In this setting algorithms that promise worst-case regret guarantees (when their recommendations are in fact played) will enjoy these guarantees for every realized sequence. In this situation, in order for agent $i$ to compute her posterior, what is relevant is not her full arrival belief $\mathcal{T}^i$, but only her marginal temporal belief $\mathcal{D}^i$ that models her belief over her own arrival-time and not that of others. Our goal in this work is to show that this simple vector of compliant strategies is an approximate Bayes Nash equilibrium, and hence in our main results and analysis, we only need to deal with (and assume agents maintain) $\mathcal{D}^i$, marginal arrival beliefs; the joint belief $\mathcal{T}^i$ is defined only to make the game well defined in other strategy profiles.
\end{remark}

\begin{definition}[Incentive-Compatible Bayes Nash Equilibrium] An algorithm implements an \emph{Incentive-Compatible Bayes Nash Equilibrium} (IC-BNE) for a set of agents with beliefs $\{(\mathcal{P}^i, \mathcal{D}^i)\}_{i=1}^T$ if, in the recommendation game, the strategy profile $\sigma^*$ (each agent follows their recommended action) is a Bayes Nash equilibrium. That is, for every agent $i \in [T]$ every alternative strategy $\sigma_i': A \to A$, and every action $a \in A$ (for which the conditioning event has non-zero probability),
\begin{equation*}
    \mathbb{E}[\mu_{\tau_i, \sigma_i^*(I_{\tau_i})} | I_{\tau_i} = a, \sigma_{-i}^*] \geq \mathbb{E}[\mu_{\tau_i,\sigma_i'(I_{\tau_i})} |  I_{\tau_i} = a, \sigma_{-i}^*]
\end{equation*}
Similarly, $\mathcal{A}$ implements an \textit{approximately} Incentive-Compatible Bayes $\epsilon$-Nash Equilibrium if
\begin{equation*}
    \mathbb{E}[\mu_{\tau_i, \sigma_i^*(I_{\tau_i})} | I_{\tau_i} = a, \sigma_{-i}^*] \geq \mathbb{E}[\mu_{\tau_i,\sigma_i'(I_{\tau_i})} |  I_{\tau_i} = a, \sigma_{-i}^*] - \epsilon
\end{equation*}
\end{definition}

\section{Regret Gives Approximate Incentive-Compatibility}
\subsection{Preliminaries}
At a high level, both regret and incentive-compatibility impose similar constraints. Exactly incentive-compatible recommendations must each be optimal in expectation, conditional on the agent's available information --- this is a point-wise requirement. A swap regret bound of zero is a morally similar constraint, marginalized over time --- it requires that on average,  the sequence of recommended actions performs as well as any fixed alternative, even conditional on the recommendation. Thus, swap regret controls the extent to which incentive constraints can be violated in aggregate as marginalized over time. 

The central focus of this section will be to bridge the gap between the marginal guarantees of regret and the pointwise incentive constraints required for equilibrium. When an agent knows their exact position in sequence, of course a marginal guarantee need not imply anything about the optimality of their specific recommendation. However, when agents are uncertain about arrival time, their individual incentive constraints are also marginalized over time, as weighted by their belief about their arrival distribution. So, if we can provide appropriately weighted regret bounds these should automatically guarantee approximate incentive compatibility.

Real algorithms will not obtain zero swap regret, but will be able to obtain diminishing weighted swap regret for sufficiently diffuse arrival beliefs, which can be made arbitrarily small with a long enough time horizon. This motivates our focus on approximate incentive-compatibility, where the approximation error is a function of worst case regret bounds. We show that in our adversarial setting, under a ``slowly moving'' assumption, algorithms achieving sublinear belief-weighted external regret imply belief-weighted swap regret and implement incentive-compatible Bayes $\epsilon$-Nash equilibria. Since this is defined relative to agent beliefs, any such result necessarily depends on mild conditions over those beliefs, similar to those required in prior work \citep{MansourSlivkins_BIC}. At a high level, we need every action to be \textit{explorable}: informally, the agent must regard each action as plausibly optimal under some realization, as otherwise no recommendation scheme could persuade them to play that action.

\begin{assumption}
    \label{ass:adversarial1}
    Under both reward and temporal prior beliefs $\mathcal{P}$ and $\mathcal{D}$, we have:
    \begin{equation*}
        \pi_a := \mathbb{P}_{\mu \sim \mathcal{P}}(\text{gap}_{\mu, \mathcal{D}}(a) \geq \Delta) \geq \alpha
    \end{equation*}
    for some $\alpha, \Delta > 0$ and for each action $a \in A$, where $\text{gap}_{\mu, \mathcal{D}}(a) := \min_{b \neq a}\mathbb{E}_{\mathcal{D}}[\mu_{t, a} - \mu_{t,b}]$ is the expected margin by which action $a$ outperforms any other fixed action as marginalized over the arrival distribution.
\end{assumption}

Intuitively, we require that, under an agent's beliefs, each action has a non-negligible probability of outperforming all other actions by some positive margin. This condition is consistent with other conditions specified across the BIC literature. \cite{MansourSlivkins_BIC} state a sufficient condition almost identical to ours in the stochastic setting, requiring:
\begin{equation}
    \mathbb{P}[X_i^k > \tau_{\mathcal{P}}]\geq \rho_{\mathcal{P}} \tag{P1}
\end{equation}
for some positive constants $(\tau_{\mathcal{P}}, \rho_{\mathcal{P}})$, where $X_i^k$ is the expected advantage of playing action $i$ over any other action $j$, conditioned on the algorithm having observed $k$ realized rewards of every action that is a priori preferable to action $i$; \cite{CombinatorialBandit_BIC} require the same condition in the combinatorial semi-bandits setting, where actions are some subset of atomic arms. Assumption \ref{ass:adversarial1} is similar while being tailored to dynamic rewards and arrival-time uncertainty. Rather than conditioning on a certain number of observed samples, this ``observation uncertainty'' is folded into the agent's temporal prior $\mathcal{D}$, as they cannot be sure how many samples have been seen before their recommendation.

\citet{priceofic_slivkins} prove that a \textit{pairwise non-dominance} condition is both necessary and sufficient for an action $i$ to be explorable: 
\begin{equation*}
    \mathbb{P}[\mu_j < \mathbb{E}[\mu_i]] > 0 \tag{P2}
\end{equation*}
for all $j \neq i$. Qualitatively, this is a useful condition to understand the nature of what priors are amenable to persuasion, and is a weaker condition than ours; however, we note that their finite-sample guarantees and bounds on regret are given in terms of parameters that quantify the ``gap'' in expected performance in the same way our parameters $(\alpha, \Delta)$ do. 

We will also assume that agents believe that the rewards are ``slowly changing'' with time.
\begin{assumption}
    \label{ass:slowmoving}
    Under reward prior $\mathcal{P}$, for each realization $\mu \sim \mathcal{P}$ in the support of $\mathcal{P}$, we have:
    \begin{equation*}
        ||\mu_{t+1} - \mu_t||_{\infty} \leq \rho
    \end{equation*}
    for each $t \in [T-1]$.
\end{assumption}
Prior work studies stochastic/stationary settings in which $\rho = 0$, and so this assumption is a strict generalization. 

\subsection{Main Theorems}
Recall that incentive-compatibility requires controlling each agent's expected gain from deviating from their recommended action conditional on being given that recommendation. In contrast, external regret bounds give us bounds on these deviations on average. The first key observation is that deviations conditioned on a fixed recommendation are naturally captured by \textit{swap regret}, which gives bounds not just in aggregate over the full sequence of recommendations, but over the subsequences of rounds corresponding to a fixed recommendation. This makes swap-regret naturally aligned with incentive constraints. 

The proof linking swap regret to incentive-compatibility has two main steps. First, we use swap-regret to bound the agent's expected gain from from deviating to an alternate action (unconditioned on their recommendation). Second, under Assumptions \ref{ass:adversarial1} and \ref{ass:slowmoving}, we show that an \textit{external regret} bound (which is directly implied by swap-regret) makes the agent \textit{believe} that each action is recommended with non-negligible probability. These two lemmas together let us bound the agent's \textit{conditional} expected gain for deviating. Our bounds also depend on quantitative measures of the agent's uncertainty about arrival time:
\begin{definition}[Max Dispersion]
The \emph{max dispersion} of a temporal prior $\mathcal{D} \in \Delta([T])$ is
\begin{equation*}
\Psi_{\max}(\mathcal D):=\max_{t\in I}\Psi_{\mathcal D}(t),
\end{equation*}
where we define the one-sided dispersion $\Psi_{\mathcal{D}}: [T] \to \mathbb{R}$ as $\Psi_{\mathcal{D}}(t) := \mathbb{E}_{s\sim\mathcal{D}}\left[\,|t-s|\,\right]$, and $I$ is the smallest contiguous set enclosing all points in the support of $\mathcal{D}$, i.e. $$I = \{t \in [T]: \min(\supp(\mathcal{D})) \leq t \leq \max(\supp(\mathcal{D}))\}$$
\end{definition}

\begin{definition}[Mean Dispersion]
The \emph{mean dispersion} of a temporal prior $\mathcal{D}$ is
\[
\Phi(\mathcal{D}) := \mathbb{E}_{s,t \sim \mathcal{D}}[|t-s|]
\]
\end{definition}

Both of these metrics capture some form of variance / how uncertain an agent is about arrival time. Intuitively, when an agent's arrival-time belief is highly concentrated, aggregate performance guarantees provide little information about the optimality of their specific recommendation. When the belief is sufficiently spread out, however, these guarantees become more informative. We formalize this notion using the above parameters, which appear naturally in the analysis of bounding an agent's belief over the probability of any specific action being recommended to them.  With these definitions in place, we can now state our main theorems informally: 

\begin{theorem*}[Swap-Regret$\implies$IC-BNE]
 Fix a collection of agents with reward and temporal beliefs $\{(\mathcal{P}^i, \mathcal{D}^i)\}_{i=1}^T$ that all satisfy Assumptions \ref{ass:adversarial1} and \ref{ass:slowmoving} with parameters $(\alpha, \Delta, \rho)$. For any distribution $\mathcal{D}$, define the concentration term $c(\mathcal{D}) := \sqrt{||\mathcal{D}||_2^2\ln(K)}$, and the dispersion term $d(\mathcal{D})$, where $d$ is either the function $\Psi_{\text{max}}$ or $\Phi$. Fix an algorithm $\mathcal{A}$ guaranteeing expected $\mathcal{D}^i$-weighted swap regret $r(\mathcal{D}^i)$ for all $i \in [T]$, If the effective gap dominates error terms:
 \begin{equation*}
     \Delta > r(\mathcal{D}^i) + c(\mathcal{D}^i) + \rho\cdot d(\mathcal{D}^i)
 \end{equation*}
 then $\mathcal{A}$ implements an incentive-compatible $\epsilon$-Bayes Nash equilibrium, where the approximation parameter for agent $i$ is
 \begin{equation*}
     \epsilon_i = O\left(\frac{r(\mathcal{D}^i)}{\alpha(\Delta - r(\mathcal{D}^i) - c(\mathcal{D}^i) - \rho\cdot d(\mathcal{D}^i))}\right)
 \end{equation*}
\end{theorem*}

This theorem gives a direct relationship between our notion of incentive-compatibility through approximate Nash equilibria and belief-weighted swap-regret, and importantly doesn't require any assumptions on the true realization of rewards over time (agents simply need to believe that rewards are slow-moving). However, swap-regret guarantees are more difficult to obtain efficiently, particularly in bandit settings. In slowly-moving, approximately stochastic environments, however, external regret closely approximates swap regret. Thus, when agents \textit{believe} they have sufficiently slowly-moving reward sequences, we can replace swap regret with the more tractable notion of external regret, enabling a broader class of algorithms to satisfy the guarantees that imply incentive-compatibility. 

\begin{theorem*}[External Regret$\implies$IC-BNE]
 Fix a collection of agents with reward and temporal beliefs $\{(\mathcal{P}^i, \mathcal{D}^i)\}_{i=1}^T$ that all satisfy Assumptions \ref{ass:adversarial1} and \ref{ass:slowmoving} with parameters $(\alpha, \Delta, \rho)$. Define concentration and dispersion terms $c(\mathcal{D})$ and $d(\mathcal{D})$ as in the previous theorem. Consider an algorithm $\mathcal{A}$ that guarantees expected $\mathcal{D}^i$-weighted external regret $r(\mathcal{D}^i)$ for all $i \in [T]$, If the effective gap dominates these error terms:
 \begin{equation*}
     \Delta > r(\mathcal{D}^i) + c(\mathcal{D}^i) + \rho\cdot d(\mathcal{D}^i)
 \end{equation*}
 then $\mathcal{A}$ implements an incentive-compatible $\epsilon$-Bayes Nash equilibrium, where the error term for agent $i$ is approximately:
 \begin{equation*}
     \epsilon_i = O\left(\frac{r(\mathcal{D}^i) + c(\mathcal{D}^i) + \rho\cdot d(\mathcal{D}^i)}{\alpha(\Delta - r(\mathcal{D}^i) - c(\mathcal{D}^i) - \rho\cdot d(\mathcal{D}^i))}\right)
 \end{equation*}
\end{theorem*}
Notice that in the stationary environment studied in previous work $(\rho = 0)$, in both theorems the dispersion terms vanish, and achieving small $\epsilon$ is (approximately) close to achieving small regret (as long as arrival-time belief is not too tightly concentrated anywhere). That is, regret minimization directly serves as an imperative for approximate obedience. 

\subsection{Auxiliary Lemmas}
\label{subsec:lemmas}
We now state and prove several auxiliary lemmas used in the proofs for our main theorems. We begin with a simple inequality that bounds the per-round expected difference between playing two actions to the expected gap over a temporal belief:

\begin{restatable}{lemma}{lemmaone}
\label{lem:ubound-general}
Fix a sequence of reward mean vectors $\mu = (\mu_1, \mu_2, \cdots, \mu_T)$ that satisfies $||\mu_{t+1} - \mu_t||_{\infty} \leq \rho$ for each $t \in [T-1]$. Then for any two actions $a, b \in A$, any time-step $t \in [T]$, and any temporal belief $\mathcal{D} \in \Delta([T])$, 
\begin{equation*}
    \mu_{t,a} - \mu_{t,b}
    \;\ge\;
    \mathrm{gap}_{\mu, \mathcal{D}}(a)\;-\;2\rho\,\Psi_{\mathcal{D}}(t).
\end{equation*}
\end{restatable}


The proof is straightforward and we relegate it to Appendix \ref{sec:proofs}. Next, we relate the guarantee of $\mathcal{D}$-weighted \textit{realized} regret (the quantity controlled for in standard adversarial bandit settings) to the $\mathcal{D}$-weighted \textit{pseudo-regret} guarantee that is relevant from the agent's perspective under their belief model. When agents evaluate recommendations, they are doing so using their expected rewards, whereas regret bounds are typically stated in terms of realized rewards. The following lemma shows that under the standard model, expected realized regret also gives us a bound on expected pseudo-regret. 

\begin{lemma}
\label{lem:pseudo-from-realized}
Fix an agent with temporal belief $\mathcal{D}$. Let $K:=|A|$ and define $W_2(\mathcal D):=\sum_{t=1}^T \mathcal{D}(t)^2$. For any realized mean-reward vector $\mu$, let $\{\tilde{u}_t\}_{t=1}^T$ be a reward sequence generated according to $\mu$, and let $\tilde{\Pi}_T = \{(\tilde{I}_t, \tilde{a}_t, \tilde{u}_t)\}_{t=1}^T$ be the \emph{hypothetical} transcript generated by running an algorithm $\mathcal{A}$ with this reward sequence. Then, for any confidence parameter $\delta \in(0,1)$, with probability at least $1-\delta$,
\begin{equation*}
\mathrm{PReg}_{\mathcal D}(\tilde{\Pi}_T, \mu) \;\le\;
\mathrm{Reg}_{\mathcal D}(\tilde{\Pi}_T) \;+\;
2\sqrt{2\,W_2(\mathcal D)\,\ln\!\Bigl(\frac{2K}{\delta}\Bigr)}.
\end{equation*}
and consequently,
\begin{equation*}
\mathbb{E}\bigl[\mathrm{PReg}_{\mathcal D}(\tilde{\Pi}_T, \mu)\bigr] \;\le\;
\mathbb{E}\bigl[\mathrm{Reg}_{\mathcal D}(\tilde{\Pi}_T) \bigr] \;+\;
2\sqrt{2\,W_2(\mathcal D)\,\ln(2K)}.
\end{equation*}
where expectation is taken over the algorithm's randomness.
\end{lemma}
\begin{proof}
In this proof, we are considering how the agent \textit{reasons} about realized regret based on their reward beliefs. To that end, we first fix a realization $\mu$ of mean-reward vectors (that parametrize reward distributions) and consider a reward sequence $\{\tilde{u}_t\}_{t=1}^T$ generated according to $\mu$. Note that this sequence may be \textit{completely unrelated} to the true sequence of rewards recorded by the transcript --- these represent rewards as modeled by the agent. Then, the algorithm $\mathcal{A}$ generates a (random) transcript $\tilde{\Pi}_T = \{(\tilde{I}_t, \tilde{a}_t, \tilde{u}_t)\}_{t=1}^T$ according to these modeled rewards. For all $t \in [T]$, let $\mathcal{F}_{t-1}$ be the $\sigma$-field generated by the history up to time $t-1$ (when running $\mathcal{A}$ on the simulated reward sequence $\{\tilde{u}_t\}_{t=1}^T$). Fix any $a \in A$ and define the random variable:
\begin{equation*}
M_t(a) \;:=\; \sum_{s=1}^t \mathcal{D}(s)\Bigl(\bigl(\tilde{u}_{s,a}-\mu_{s,a}\bigr)-\bigl(\tilde{u}_{s,\tilde{I}_s}-\mu_{s,\tilde{I}_s}\bigr)\Bigr).
\end{equation*}
for each $t \in [0,T]$. Since rewards are bounded in $[0,1]$, we have
\begin{equation*}
\bigl|M_t(a)-M_{t-1}(a)\bigr| \;\le\; \mathcal{D}(t)\bigl(|\tilde{u}_{t,a}-\mu_{t,a}|+|\tilde{u}_{t,\tilde{I}_t}-\mu_{t,\tilde{I}_t}|\bigr) \;\le\; 2\mathcal{D}(t).
\end{equation*}

For a fixed $\mu$, $\tilde{u}_t$ is independent of $\mathcal{F}_{t-1}$, and so $\mathbb{E}[\tilde{u}_{t, a} - \mu_{t, a} |\mathcal{F}_{t-1}] = 0$. Moreover, because $\tilde{I}_t$ is $\mathcal F_{t-1}$-measurable, $\mathbb{E}[\tilde{u}_{t, \tilde{I}_t} - \mu_{t, \tilde{I}_t} |\mathcal{F}_{t-1}] = 0$ , and so,
\begin{equation*}
\mathbb{E}\!\left[M_t(a)-M_{t-1}(a)\mid \mathcal F_{t-1}\right]=0,
\end{equation*}
which verifies that $\{M_t(a)\}_{t=0}^T$ is a martingale. By Azuma's inequality, for any $x>0$,
\begin{equation*}
\mathbb{P}\bigl(|M_T(a)|\ge x\bigr) \;\le\;
2\exp\!\left(-\frac{x^2}{8\sum_{t=1}^T \mathcal{D}(t)^2}\right) \;=\;
2\exp\!\left(-\frac{x^2}{8W_2(\mathcal D)}\right).
\end{equation*}
Applying a union bound over $a\in  A$, we get that with probability at least $1-\delta$,
\begin{equation*}
\max_{a\in A}|M_T(a)| \;\le\; 2\sqrt{2\,W_2(\mathcal D)\,\ln\!\Bigl(\frac{2K}{\delta}\Bigr)}.
\end{equation*}
By definition,
\begin{equation*}
\sum_{t=1}^T \mathcal{D}(t)\bigl(\mu_{t,a}-\mu_{t,\tilde{I}_t}\bigr) = \sum_{t=1}^T \mathcal{D}(t)\bigl(\tilde{u}_{t,a}-\tilde{u}_{t,\tilde{I}_t}\bigr) - M_T(a).
\end{equation*}
and taking the max over all $a$:
\begin{equation}
\label{eq:pseudoreg}
\mathrm{PReg}_{\mathcal D}(\mu) \;\le\; \mathrm{Reg}_{\mathcal D}(\tilde{\Pi}_T)+\max_{a}|M_T(a)| \leq \mathrm{Reg}_{\mathcal D}(\tilde{\Pi}_T) +  2\sqrt{2\,W_2(\mathcal D)\,\ln\!\Bigl(\frac{2K}{\delta}\Bigr)}. 
\end{equation}
with probability at least $1 - \delta$. For the expectation bound, we use the tail-sum formula: 
\begin{align*}
\mathbb{E}\!\left[\max_{a}|M_T(a)|\right] &\;=\;
\int_0^\infty \mathbb{P}\!\left(\max_{a}|M_T(a)|\ge x\right)\,dx \\
&\;\le\;
\int_0^\infty \min\!\left\{1,\;2K\exp\!\left(-\frac{x^2}{8W_2(\mathcal D)}\right)\right\}\,dx
\;\le\;
2\sqrt{2\,W_2(\mathcal D)\,\ln(2K)}.
\end{align*}
Taking an expectation over (\ref{eq:pseudoreg}) completes the proof. 
\end{proof}

Next, we relate swap-regret to pseudoregret. 
\begin{restatable}{lemma}{lemmatwo}
\label{lem:pseudo-from-swap}
    Fix an agent with temporal belief $\mathcal{D}$, and a realized mean-reward vector $\mu = (\mu_1, \cdots, \mu_T)$. Let $\{\tilde{u}_t\}_{t=1}^T$ be a (random) reward sequence generated according to $\mu$, and let $\tilde{\Pi}_T = \{(\tilde{I}_t, \tilde{a}_t, \tilde{u}_t)\}_{t=1}^T$ be the \emph{hypothetical} transcript generated by running an algorithm $\mathcal{A}$ with this reward sequence. Then, for any fixed swap function $\phi: A \to A$,
    \begin{equation*}
        \mathbb{E}\left[\sum_{t=1}^T \mathcal{D}(t) (\mu_{t, \phi(\tilde{I}_t)} - \mu_{t, \tilde{I}_t})\right] =         \mathbb{E}\left[\sum_{t=1}^T \mathcal{D}(t) (\tilde{u}_{t, \phi(\tilde{I}_t)} - \tilde{u}_{t, \tilde{I}_t})\right] \leq \mathbb{E}[\textrm{SReg}_{\mathcal{D}}(\tilde{\Pi}_T)]
    \end{equation*}
\end{restatable}
Another largely straightforward proof, we defer it to Appendix \ref{sec:proofs}. We must also characterize the relationship between pseudo swap-regret and pseudo-regret.

\begin{lemma}
\label{lem:ext-to-swap}
Fix any temporal belief $\mathcal{D}$ and a reward belief $\mathcal{P}$ that satisfies Assumption \ref{ass:slowmoving} with parameter $\rho$. Consider any transcript $\Pi_T = \{(I_t, a_t, u_t)\}_{t=1}^T$ generated with rewards drawn based on $\mu \sim \mathcal{P}$. Then, 
\[
    \textrm{PSReg}_{\mathcal{D}}(\Pi_T, \mu) \leq \textrm{PReg}_{\mathcal{D}}(\Pi_T, \mu) + 2 \rho \Phi(\mathcal{D}) 
\]
and consequently, taking an expectation over $\mathcal{P}$,
\[
        \mathbb{E}[\textrm{PSReg}_{\mathcal{D}}(\Pi_T, \mu)] \leq \mathbb{E}[\textrm{PReg}_{\mathcal{D}}(\Pi_T, \mu)] + 2 \rho \Phi(\mathcal{D})
\]
\end{lemma}
\begin{proof}
Fix any $\mu \sim \mathcal{P}$. For any pair of time-steps $t, s \in [T]$, we have:
\begin{equation*}
    \mu_{t, a} - \mu_{s,a} \leq \rho|s-t|
\end{equation*}
for any $a \in A$, by Assumption \ref{ass:slowmoving}. If $s \sim \mathcal{D}$, we can take an expectation on both sides to get:
\begin{equation*}
    \mu_{t,a} \leq \mathbb{E}_{s \sim \mathcal{D}}[\mu_{s,a}] + \rho\Psi_\mathcal{D}(t)
\end{equation*}
Fix any swap function $\phi: A \to A$. We can bound the weighted pseudo-swap regret with respect to $\phi$ as:
\begin{align*}
    \sum_{t=1}^T \mathcal{D}(t) \mu_{t, \phi(I_t)} &\leq \sum_{t=1}^T \mathcal{D}(t) \mathbb{E}_{s \sim \mathcal{D}}[\mu_{s, \phi(I_t)}] + \rho\sum_{t=1}^T \mathcal{D}(t) \Psi_{\mathcal{D}}(t) \\
    &= \mathbb{E}_{s \sim \mathcal{D}}\left[\sum_{t=1}^T \mathcal{D}(t) \mu_{s, \phi(I_t)}\right] + \rho \Phi(\mathcal{D}) \\
    &\leq \mathbb{E}_{s \sim \mathcal{D}}\left[\max_{a \in A} \mu_{s, a} \right] + \rho\Phi(\mathcal{D})
\end{align*}
We can bound the final expectation by first noting for any $a \in A$ and an independent $s' \sim \mathcal{D}$ that if we fix $s$,
\begin{equation*}
    \mu_{s,a} \leq \mathbb{E}_{s' \sim \mathcal{D}}[\mu_{s', a}] + \rho\mathbb{E}_{s' \sim \mathcal{D}}|s' - s| \implies \max_{a \in A} \mu_{s,a} \leq \max_{a \in A} \mathbb{E}_{s' \sim \mathcal{D}}[\mu_{s', a}] + \rho\mathbb{E}_{s' \sim \mathcal{D}}|s' - s| 
\end{equation*}
Now, taking expectation also over $s$,
\begin{equation*}
    \mathbb{E}_{s \sim \mathcal{D}}\left[ \max_{a \in A} \mu_{s,a} \right] \leq \max_{a \in A} \mathbb{E}_{s' \sim \mathcal{D}}[\mu_{s', a}] + \rho\mathbb{E}_{s,s' \sim \mathcal{D}}|s' - s| = \max_{a \in A} \sum_{t=1}^T \mathcal{D}(t) \mu_{t, a} + \rho\Phi(\mathcal{D})
\end{equation*}
and so
\begin{equation*}
    \sum_{t=1}^T \mathcal{D}(t) \mu_{t, \phi(I_t)} \leq \max_{a \in A} \sum_{t=1}^T\mathcal{D}(t) \mu_{t, a} + 2\rho\Phi(\mathcal{D})
\end{equation*}
This is true for all $\phi \in \Phi$, so subtracting $\sum_{t=1}^T \mathcal{D}(t) \mu_{t,I_t}$ from both sides gives:
\begin{equation*}
    \textrm{PSReg}_{\mathcal{D}}(\Pi_T, \mu) \leq \textrm{PReg}_{\mathcal{D}}(\Pi_T, \mu) + 2\rho\Phi(\mathcal{D})
\end{equation*}
This holds for any realization of transcript; taking expectation over algorithm randomness and $\mathcal{P}$ gives the same bound in expectation.
\end{proof}

Finally, we use these lemmas to get bounds on the term $\Pr(I_{\tau_i} = a | \sigma_{-i})$ that we use in Theorems \ref{thm:BIC-swap} and \ref{thm:BIC-ext}. These bounds are the main ones needed to convert regret guarantees to incentive-compatibility ones. Given an algorithm achieves no external-regret, an agent with the appropriate beliefs will believe that there is a non-negligible probability of any action being recommended. We obtain two separate but similar bounds, in terms of our two dispersion parameters $\Psi_{\max}$ and $\Phi$.
\begin{lemma}
\label{lem:prob-bound-max1}
Fix an agent $i$ with prior distributions $\mathcal{P}^i, \mathcal{D}^i$ that satisfy Assumptions \ref{ass:adversarial1} and \ref{ass:slowmoving} with parameters $(\alpha, \Delta, \rho)$. Consider play of the recommendation game with some algorithm $\mathcal{A}$, where all agents follow recommendations (the compliant strategy profile $\sigma^*$), and $i$ has observed their recommendation $I_{\tau_i}$. Set $\tilde\Delta:=\Delta-2\rho\,\Psi_{\max}(\mathcal D^i)$.  If $\tilde\Delta>0$, then any algorithm $\mathcal A$ with expected
$\mathcal D^i$-weighted external regret $R^{(\mathcal D^i)}(T)$ satisfies
\begin{equation*}
\Pr\left(I_{\tau_i} = a \right|\sigma^*_{-i}\big)
\;\ge\;
\alpha\left(1-\frac{R^{(\mathcal D^i)}(T) + 2\sqrt{2\,W_2(\mathcal D^i)\,\ln(2K)}}{\tilde\Delta}\right)
\end{equation*}
for each action $a \in A$. 
\end{lemma}

\begin{proof}
Fix an action $a \in A$, and define the corresponding event $E_a := \{\mu: \text{gap}_{\mu, \mathcal{D}^i}(a) \geq \Delta \}$. Fix a reward instance $\mu \in E_a$, and let $\mathcal{A}$ generate the hypothetical transcript $\tilde{\Pi}_T = \{(\tilde{I}_t, \tilde{a}_t, \tilde{u}_t)\}_{t=1}^T$ under rewards generated by $\mu$. Define $I = \{t \in [T]: \min(\supp(\mathcal{D}^i)) \leq t \leq \max(\supp(\mathcal{D}^i))\}$. By Lemma \ref{lem:ubound-general}, for any time-step $t \in I$ and action $b \neq a$,
\begin{equation*}
    \mu_{t,a} - \mu_{t,b} \geq \text{gap}_{\mu, \mathcal{D}^i}(a) - 2\rho\Psi_{\mathcal{D}^i}(t) \geq \Delta - 2\rho\Psi_{\max}(\mathcal{D}^i) = \tilde{\Delta}
\end{equation*}
Hence, for all $t \in I$,
\[
\mu_{t, a} - \mu_{t, \tilde{I}_t}\ \ge\ \tilde\Delta\,\mathbf 1[\tilde{I}_t\ne a]
\]
Note here that $\tilde{I}_t$ refers to the recommendation issued by the transcript, not the private observation $I_{\tau_i}$ of the agent. Taking sums on both sides weighted by $\mathcal{D}^i$ we get:
\begin{equation*}
    \sum_{t \in I} \mathcal{D}^i(t) \left(\mu_{t, a} - \mu_{t, \tilde{I}_t}\right) \geq \tilde{\Delta}\sum_{t \in I} \mathcal{D}^i(t) \mathbf 1[\tilde{I}_t\ne a] = \tilde{\Delta} (1 - N_a)
\end{equation*}
where $N_a := \sum_{t \in I} \mathcal{D}^i(t) \mathbf 1[\tilde{I}_t = a]$ is by definition the probability under $\mathcal{D}^i$ of a recommended action being $a$, for a fixed $\mu$. Therefore, $\mathbb{E}[N_a] = \Pr(I_{\tau_i} = a | \sigma_{-i}^*)$, where the expectation is being taken over $\mathcal{P}^i$ and the algorithm's randomness. Note that the left-most term above is pseudo-regret against the fixed comparator action $a$, so
\begin{equation*}
    \mathrm{PReg}_{\mathcal D^i}(\tilde{\Pi}_T, \mu) \geq \tilde{\Delta}(1 - N_a) \implies N_a \geq 1 - \frac{\mathrm{PReg}_{\mathcal D^i}(\tilde{\Pi}_T, \mu)}{\tilde{\Delta}}
\end{equation*}
Taking an expectation over the algorithm's randomness and using Lemma \ref{lem:pseudo-from-realized}, 
\begin{equation*}
    \mathbb{E}_{\mathcal{A}}[N_a] \geq 1 - \frac{\mathbb{E}_{\mathcal{A}}\left[\mathrm{Reg}_{\mathcal{D}^i}(\tilde{\Pi}_T)\right] + 2\sqrt{2\,W_2(\mathcal D^i)\,\ln(2K)}}{\tilde{\Delta}} \geq  1 - \frac{\R(T) + 2\sqrt{2\,W_2(\mathcal D^i)\,\ln(2K)}}{\tilde{\Delta}}
\end{equation*}
with the final step coming from the regret guarantee of $\mathcal{A}$. The above inequality is true for any $\mu \in E_a$. In general, $N_a \geq 0$ by definition, so $\mathbb{E}[N_a |\; E_a^C] \geq 0$. Further, recall that by Assumption \ref{ass:adversarial1}, $\Pr(E_a) \geq \alpha$. So, taking an expectation over the full prior $\mathcal{P}^i$,
\begin{align*}
    \mathbb{E}[N_a] \geq \Pr(E_a)\;\mathbb{E}[N_a | E_a] &\geq \alpha\left(1 - \frac{R^{(\mathcal{D}^i)}(T) + 2\sqrt{2\,W_2(\mathcal D^i)\,\ln(2K)}}{\tilde{\Delta}}\right)
\end{align*}
\end{proof}

The following lemma's proof proceeds nearly identically to that of Lemma \ref{lem:prob-bound-max1}, so we leave the details in Appendix \ref{sec:proofs}. 

\begin{restatable}{lemma}{lemmathree}
\label{lem:prob-bound-max2}
Fix an agent $i$ with prior distributions $\mathcal{P}^i, \mathcal{D}^i$ that satisfy Assumptions \ref{ass:adversarial1} and \ref{ass:slowmoving} with parameters $(\alpha, \Delta, \rho)$. Consider play of the recommendation game with some algorithm $\mathcal{A}$, where all agents follow recommendations (the compliant strategy profile $\sigma^*$), and $i$ has observed their recommendation $I_{\tau_i}$. Then any algorithm $\mathcal A$ with expected $\mathcal D^i$-weighted external regret $R^{(\mathcal D^i)}(T)$ satisfies
\begin{equation*}
\Pr\left(I_{\tau_i} = a \right|\sigma^*_{-i}\big)
\;\ge\;
\alpha\left(1-\frac{\R(T)+ 2\sqrt{2\,W_2(\mathcal D^i)\,\ln(2K)} + 2\rho\Phi(\mathcal{D}^i)}{\Delta}\right)
\end{equation*}
for each action $a \in A$. 
\end{restatable}

\subsection{Proofs of Main Theorems}
Having  built up our machinery, we formally restate our theorems linking forms of regret to approximate incentive-compatibility and provide their proofs. 

\begin{theorem}[Swap-Regret$\implies$IC-BNE]
\label{thm:BIC-swap}
Fix a collection of agents with reward and temporal beliefs $\{(\mathcal{P}^i, \mathcal{D}^i)\}_{i=1}^T$ that all satisfy Assumptions \ref{ass:adversarial1} and \ref{ass:slowmoving} with parameters $(\alpha, \Delta, \rho)$. If an algorithm $\mathcal{A}$ played over $T$ rounds guarantees expected $\mathcal{D}^i$-weighted swap regret $R_{\emph{Swap}}^{(\mathcal{D}^i)}(T)$ for all $i \in [T]$, then $\mathcal{A}$ implements an incentive-compatible Bayes $\epsilon(T)$-Nash equilibrium with error term for agent $i$:
\[
\epsilon_i(T) \;=\; \min\{\eta_{\Psi}(T),\, \eta_{\Phi}(T), 1\}.
\]
where we define
    \begin{equation*}
    \eta_{\Psi}(T)
    =\frac{\RSwap(T)\,
    (\Delta-2\rho\,\Psi_{\max}(\mathcal D^i))}
          {\alpha\left((\Delta-2\rho\,\Psi_{\max}(\mathcal D^i))
          -\RSwap(T) - 2\sqrt{2\,W_2(\mathcal D^i)\,\ln(2K)}\right)}
    \end{equation*}
    \begin{equation*}
    \eta_{\Phi}(T)
    =\frac{\RSwap(T)\Delta}
          {\alpha(\Delta - (\RSwap(T) + 2\sqrt{2\,W_2(\mathcal D^i)\,\ln(2K)} + 2\rho\Phi(\mathcal{D}^i)))}
    \end{equation*}
    where both $\eta_{\Psi}(T)$ and $\eta_{\Phi}(T)$ are well-defined only if they are positive. Otherwise they are interpreted as $+\infty$.
\end{theorem}
\begin{proof}
Say $\mathcal{A}$ has generated the (random) transcript $\Pi_T = \{(I_t, a_t, u_t)\}_{t=1}^T$. Fix the compliant strategy profile $\sigma^* = (\sigma_1^*, \sigma_2^*, \cdots, \sigma_T^*)$ in which each agent follows their recommended action. Suppose agent $i$ has been given recommendation $I_{\tau_i} = a$. A deviation from strategy corresponds to choosing a different recommendation-to-action function, and since all other agent strategies are fixed, it suffices to consider deviations from the given recommendation $a$ to any other action $b \neq a$. Conditioning on the given recommendation,
\begin{equation}
\label{eq:imp-eq}
\mathbb{E}[\mu_{\tau_i,b}-\mu_{\tau_i,a}\mid I_{\tau_i}=a, \sigma_{-i}^*]
=\frac{\mathbb{E}\left[(\mu_{\tau_i,b}-\mu_{\tau_i,a})\mathbf{1}[I_{\tau_i}=a] | \sigma_{-i}^*\right]}
       {\Pr(I_{\tau_i}=a | \sigma_{-i})}.
\end{equation}
where all expectations and probabilities are taken over $\mu \sim \mathcal{P}^i, \tau_i \sim \mathcal{D}^i$ and the agent's knowledge of $\mathcal{A}$. Taking the expectation with respect to $\mathcal{D}^i$ into the expression, we get:
\begin{align*}
\mathbb{E}\!\left[(\mu_{\tau_{i},b}-\mu_{\tau_i,a})\mathbf{1}[I_{\tau_i}=a] | \sigma_{-i}^* \right]
&= \mathbb{E}\!\left[\sum_{t=1}^T
\mathcal D^i(t)\,(\mu_{t,b} -\mu_{t,a})\,\mathbf{1}[I_{t}=a] \Big| \sigma_{-i}^*\right] \\&= \mathbb{E}\!\left[\sum_{t=1}^T
\mathcal D^i(t)\,(\mu_{t,\phi(I_t)} -\mu_{t,I_t})\,\Big| \sigma_{-i}^*\right] \leq \RSwap(T)
\end{align*}
using Lemma \ref{lem:pseudo-from-swap} with the swap function $\phi(a) = b,$ and $\phi(k) = k \; \forall k \neq a$. Note that the regret guarantee holds since we are conditioning on all agents following their given recommendations. We can lower bound $1 / \Pr(I_{\tau_i}=a | \sigma_{-i})$ two ways, directly using Lemmas \ref{lem:prob-bound-max1} and \ref{lem:prob-bound-max2}.
\begin{equation*}
    \frac{1}{\Pr(I_{\tau_i}=a | \sigma_{-i})} \leq \frac{\Delta - 2\rho\Psi_{\max}(\mathcal{D}^i)}{\alpha(\Delta - 2\rho\Psi_{\max}(\mathcal{D}^i) - R_{\text{Swap}}^{(\mathcal{D}^i)}(T)- 2\sqrt{2\,W_2(\mathcal D^i)\,\ln(2K)})} 
\end{equation*}
and
\begin{equation*}
    \frac{1}{\Pr(I_{\tau_i}=a|\sigma_{-i})} \leq \frac{\Delta}{\alpha(\Delta - R_{\text{Swap}}^{(\mathcal{D}^i)}(T) - 2\sqrt{2\,W_2(\mathcal D^i)\,\ln(2K)} - 2\rho\Phi(\mathcal{D}^i))}
\end{equation*}
where each inequality holds iff the right-hand-side is positive. Note that these lemmas are stated in terms of external regret, and we can use swap-regret because it upper-bounds external regret. Substituting these inequalities back into equation (\ref{eq:imp-eq}), we get:
\begin{equation*}
\mathbb{E}[\mu_{\tau_i,b}-\mu_{\tau_i,a}\mid I_{\tau_i}=a, \sigma_{-i}^*] \leq  \min\{\eta_{\Psi}(T),\, \eta_{\Phi}(T)\}
\end{equation*}
as desired.
\end{proof}

\begin{theorem}
\label{thm:BIC-ext}
Fix a collection of agents with reward and temporal beliefs $\{(\mathcal{P}^i, \mathcal{D}^i)\}_{i=1}^T$ that all satisfy Assumptions \ref{ass:adversarial1} and \ref{ass:slowmoving} with parameters $(\alpha, \Delta, \rho)$. If an algorithm $\mathcal{A}$ played over $T$ rounds guarantees expected $\mathcal{D}^i$-weighted external regret $R^{(\mathcal{D}^i)}(T)$ for all $i \in [T]$, then $\mathcal{A}$ implements an incentive-compatible Bayes $\epsilon(T)$-Nash equilibrium with error term for agent $i$:
\[
\epsilon_i(T) \;=\; \min\{\eta_{\Psi}(T),\, \eta_{\Phi}(T), 1\}.
\]
where we define
    \begin{equation*}
    \eta_{\Psi}(T)
    =\frac{(\R(T) + 2\rho\Phi(\mathcal{D}^i) +2\sqrt{2W_2(\mathcal{D}^i)\ln(2K)})\,
    (\Delta-2\rho\,\Psi_{\max}(\mathcal D^i))}
          {\alpha(\Delta - 2\rho\Psi_{\max}(\mathcal{D}^i) - \R(T) - 2\sqrt{2\,W_2(\mathcal D^i)\,\ln(2K)})}
    \end{equation*}
    \begin{equation*}
    \eta_{\Phi}(T)
    =\frac{(\R(T) + 2\rho\Phi(\mathcal{D}^i) + 2\sqrt{2W_2(\mathcal{D}^i)\ln(2K)})\Delta}
          {\alpha(\Delta - 2\rho\Phi(\mathcal{D}^i) - \R(T) - 2\sqrt{2\,W_2(\mathcal D^i)\,\ln(2K)})}
    \end{equation*}
Again, both $\eta_{\Psi}(T)$ and $\eta_{\Phi}(T)$ are well-defined only if positive. Otherwise, they are interpreted to be $+\infty$. 
\end{theorem}
\begin{proof}
The proof follows closely that of Theorem \ref{thm:BIC-swap}. Let $\mathcal{A}$ generate a random transcript $\Pi_T = \{(I_t, a_t, u_t)\}_{t=1}^T$ and consider the compliant strategy profile $\sigma^*$ where all agents follow their recommendations. Consider any agent $i$ who receives recommendation $I_{\tau_i} = a$, and any strategy deviation which causes them to take action $b \neq a$.
\begin{equation}
\label{eq:imp-eq-repeated}
\mathbb{E}[\mu_{\tau_i,b}-\mu_{\tau_i,a}\mid I_{\tau_i}=a, \sigma_{-i}^*]
=\frac{\mathbb{E}\left[(\mu_{\tau_i,b}-\mu_{\tau_i,a})\mathbf{1}[I_{\tau_i}=a] | \sigma_{-i}^*\right]}
       {\Pr(I_{\tau_i}=a | \sigma_{-i})}.
\end{equation}
Note that the numerator is upper-bounded by pseudo-swap regret, so:
\begin{align*}
    \mathbb{E}\left[(\mu_{\tau_i,b}-\mu_{\tau_i,a})\mathbf{1}[I_{\tau_i}=a] | \sigma_{-i}^*\right] &\leq \mathbb{E}[\textrm{PSReg}_{\mathcal{D}^i}(\Pi_T, \mu)] \\
    &\leq \mathbb{E}[\textrm{PReg}_{\mathcal{D}^i}(\Pi_T, \mu)] + 2\rho\Phi(\mathcal{D}^i) \\
    &\leq \R(T) + 2\rho\Phi(\mathcal{D}^i) + 2\sqrt{2W_2(\mathcal{D}^i)\ln(2K)}
\end{align*}
using Lemmas \ref{lem:pseudo-from-realized} and \ref{lem:ext-to-swap}.
We lower-bound $1 / \Pr(I_{\tau_i}=a)$ two ways exactly as in Theorem \ref{thm:BIC-swap},
\begin{equation*}
    \frac{1}{\Pr(I_{\tau_i}=a | \sigma_{-i})} \leq \frac{\Delta - 2\rho\Psi_{\max}(\mathcal{D}^i)}{\alpha(\Delta - 2\rho\Psi_{\max}(\mathcal{D}^i) - \R(T) - 2\sqrt{2\,W_2(\mathcal D^i)\,\ln(2K)})} 
\end{equation*}
\begin{equation*}
    \frac{1}{\Pr(I_{\tau_i}=a)} \leq \frac{\Delta}{\alpha(\Delta - (\R(T) + 2\sqrt{2\,W_2(\mathcal D^i)\,\ln(2K)} + 2\rho\Phi(\mathcal{D}^i))}
\end{equation*}
noting that Lemmas \ref{lem:prob-bound-max1} and \ref{lem:prob-bound-max2} relied only on an external regret guarantee. Substituting these inequalities back into equation (\ref{eq:imp-eq-repeated}) gives the stated result.
\end{proof}

\section{Concrete Bounds for Structured Temporal Beliefs}
The general statement of Theorem \ref{thm:BIC-ext} is intentionally broad, but this generality makes the resulting bound somewhat abstract. In particular, both the dispersion and regret terms have a dependence on the agent's temporal belief distribution $\mathcal{D}$, and the regret bound may vary significantly based on the specific algorithm $\mathcal{A}$ used to achieve it. Thus it is not immediately clear when the bound is strong or how it behaves in concrete settings. To obtain better insight, we focus on structured distributions $\mathcal{D}$ for which terms like $\Psi_{\max}(\mathcal{D})$ and $\Phi(\mathcal{D})$ reduce to explicit, easily computable formulas and for which strong regret bounds are attainable. 

One key observation we make first is that incentive-compatibility is a property defined through conditional expectations, and behaves linearly with respect to an agent's beliefs. If an agent's temporal beliefs can be expressed as a mixture of simpler distributions, the agent's expected gain from deviation is the convex combination of the corresponding gains under the constituent beliefs. That is, achieving incentive-compatibility explicitly for a finite collection of temporal belief distributions immediately extends to all distributions in their convex hull. As we will see, this result does more than extend incentive-compatibility to a larger class of beliefs; by establishing incentive-compatibility for more localized temporal beliefs, it allows us to also handle highly diffuse beliefs without imposing strong drift constraints, thus accommodating truly dynamic environments.

\begin{lemma}
\label{lem:mixtures-bic}
Fix an algorithm $\mathcal{A}$, a collection of $T$ agents, and some agent $i \in [T]$. Consider any set $\{\mathcal{D}^i_k\}_{k=1}^m$ of possible temporal beliefs for agent $i$. If $\mathcal{A}$ implements an incentive-compatible $\epsilon$-BNE for the instance in which agent $i$ has temporal belief $\mathcal{D}^i_k$ (while all other agents have unchanged beliefs), for all $k \in [m]$, then it also implements an incentive-compatible $\epsilon$-BNE for the instance in which agent $i$ has belief $\bar{\mathcal D}^i$, any distribution in the convex hull of $\{\mathcal{D}^i_k\}_{k=1}^m$, i.e.
\begin{equation*}
    \bar{\mathcal D}^i \;:=\; \sum_{k=1}^m \lambda_k \mathcal D^i_k .
\end{equation*}
for some set of non-negative mixture weights $\lambda = \{\lambda_k\}_{k=1}^m$ summing to 1.
\end{lemma}
\begin{proof}
    Defining the random variable $K$ sampled as $K \sim \lambda$ to choose a base prior $D^i_K$, and then sampling from $D^i_K$ is equivalent to sampling from the mixture prior $\bar{\mathcal D}^i$. Fix two actions $a,b \in A$ and let the random variable $X_{a,b} := \mu_{\tau_i,b} - \mu_{\tau_i,a}$. By the law of total expectation under the mixture,
\begin{equation*}
\mathbb{E}\!\left[X_{a,b} \mid I_{\tau_i} = a, \sigma_{-i}^{*}\right]
= \sum_{k=1}^m \Pr(K=k \mid I_{\tau_i} = a, \sigma_{-i}^*)\;
\mathbb{E}\!\left[X_{a,b} \mid I_{\tau_i} = a,\,  K= k, \sigma_{-i}^*\right],
\end{equation*}
where
\begin{equation*}
\Pr(K=k \mid I_{\tau_i} = a, \sigma_{-i}^*)
=
\frac{\lambda_k\,\Pr(I_{\tau_i} = a\mid K = k,\sigma_{-i}^*)}{\sum_{l=1}^m \lambda_\ell\,\Pr(I_{\tau_i}= a\mid K = l, \sigma_{-i}^*)}.
\end{equation*}
The weights $\Pr(K=k \mid I_{\tau_i} = a, \sigma_{-i}^*)$ thus are nonnegative and sum to $1$ whenever $\Pr(I_{\tau_i} = a\mid \sigma_{-i}^*)>0$; moreover, any $k$ with $\Pr(I_{\tau_i} = a\mid  K = k, \sigma_{-i}^*)=0$ receives posterior weight $0$.

By assumption, each conditional expectation $\mathbb{E}\!\left[X_{a,b} \mid I_{\tau_i} = a,\,  K= k, \sigma_{-i}^*\right]$ in the sum is at most $\epsilon$.
Therefore,
\begin{equation*}
\mathbb{E}\!\left[X_{a,b} \mid I_{\tau_i} = a, \sigma_{-i}^{*}\right] \;\le\;
\sum_{k=1}^m  \Pr(K=k \mid I_{\tau_i} = a, \sigma_{-i}^*)\;\epsilon \;=\; \epsilon,
\end{equation*}
which completes the proof.
\end{proof}

Although we stated Lemma \ref{lem:mixtures-bic} in terms of mixing the beliefs of a single agent (while keeping other agents' beliefs fixed), the same argument extends immediately to simultaneous mixing by all agents. That is, if every agent has a collection of base beliefs, and, for every combination of these beliefs, $\mathcal{A}$ gives an incentive-compatible $\epsilon$-BNE, then so does it when each agent chooses from the convex hull of their base beliefs. This observation allows us to restrict our attention and analysis to a small collection of well-behaved distributions. In this paper, we focus on agent temporal beliefs that are uniform over a contiguous subsequence of length $L$ over the full time horizon. For ease of notation, we define $\mathcal{U}_{s,L}$ as the uniform distribution over time-steps $\{s, s+1, \cdots, s+L-1\}$. With these distributions, it is straightforward to compute relevant quantities: 
\begin{equation}
\label{eq:uniform-terms}
    \Psi_{\max}(\mathcal{U}_{s,L}) = \frac{L-1}{2}, \;\;\;\;\;\;\;\Phi(\mathcal{U}_{s,L}) = \frac{L^2-1}{3L}, \;\;\;\;\;\;\; ||\mathcal{U}_{s,L}||_2^2 = \frac{1}{L} 
\end{equation}
Note that the $\mathcal{U}_{s,L}$-weighted external regret is simply the \textit{average} regret on the corresponding interval:
\begin{equation*}
    \textrm{Reg}_{\mathcal{U}_{s,L}}(\Pi_T) = \frac{1}{L} \max_{a \in A} \sum_{t=s}^{s+L-1} (u_{t,a} - u_{t,I_t})
\end{equation*}
where the summation on the right is unweighted external regret. In this setting, an interval regret bound on $\{s, \cdots, s+L-1\}$ translates to a $\mathcal{U}_{s,L}$-weighted external regret guarantee scaled down by $L$.

The problem of achieving sublinear regret bounds simultaneously on subsequences up to a particular length is a well-studied problem --- a form of regret known as \textit{adaptive regret} \citep{hazanseshadri_adaptive}. \citet{luo_intervalreg} and \citet{adaptiveregret_opt} separately derive algorithms with adaptive regret bounds in the bandit setting.
\begin{theorem}[\citet{luo_intervalreg}, Theorem 2] 
\label{thm:luo}
In the standard adversarial bandits setting with $K$ actions, Exp4.S with parameter $L$ generates a series of recommendations $\{I_t\}_{t=1}^T$ such that for any contiguous time interval $\mathcal{I}$ with $|\mathcal{I}| \leq L$, any sequence of reward vectors $\{u_t\}_{t=1}^T$, and any action $a \in A$, we have
$\mathbb{E}\left[\sum_{t \in \mathcal{I}} u_{t,a} - u_{t,I_t}\right] \leq O(\sqrt{LK\ln(LK)})$, where expectation is with respect to both the algorithm and the environment. 
\end{theorem}
Setting $L = T$ gives regret bounds on the order of $\tilde{O}(\sqrt{TK})$ uniformly across all contiguous intervals. If one is interested only in intervals smaller than some fixed $L$, the sharper rates of $\tilde{O}(\sqrt{LK})$ are achievable for intervals of length at most $L$, but not uniformly across all (larger) intervals\footnote{Note that \textit{strong} adaptive regret bounds that uniformly achieve $\tilde{O}(\sqrt{|I|})$ bounds for all intervals $I$ are not achievable in the single-query MAB setting \citep{daniely_strongadaptive}.}. We can consider both these regimes to derive concrete bounds for approximate IC using \textit{existing} algorithms with strong regret guarantees.
\begin{theorem}
\label{thm:uniform-first}
    Define the set $U = \{\mathcal{U}_{s,L}: L \in [T], s \in [T - L + 1]\}$ as the set of uniform distributions over all possible contiguous intervals over $[T]$. Fix a collection of agents with beliefs $\{(\mathcal{P}^i, \mathcal{D}^i)\}_{i=1}^T$ all of which satisfy Assumptions \ref{ass:adversarial1} and \ref{ass:slowmoving} with parameters $(\alpha, \Delta,\rho)$, and such that $\mathcal{D}^i \in U$ for all $i \in [T]$. If agent $i$'s temporal belief is uniform over a subsequence of length $L$, then Exp4.S with parameter $T$ implements an incentive-compatible Bayes $\epsilon(T)$-Nash equilibrium where the error term for agent $i$ is:
    \begin{equation*}
        \epsilon_i(T) = \frac{\left(O\left(\frac{\sqrt{TK\ln(TK)}}{L}\right) + \frac{2\rho(L^2 - 1)}{3L} + 2\sqrt{\frac{2\ln(2K)}{L}}\right)(\Delta - \rho(L-1))}{\alpha\left(\Delta - O\left(\frac{\sqrt{TK\ln(TK)}}{L}\right) - \rho(L-1) - 2\sqrt{\frac{2\ln(2K)}{L}}\right)}
    \end{equation*}
    Specifically, when $L = \tilde{\omega}(\sqrt{T})$, and $\rho = o(1/L)$, we have vanishing error $\epsilon_i(T) = o(1)$. 
\end{theorem}
\begin{proof}
Substituting in the bounds from Theorem \ref{thm:luo} and computed terms from line (\ref{eq:uniform-terms}) gives the result. It is clear to see that when $L$ grows faster than $\sqrt{T\ln(T)}$, the regret bounds go to zero, and when $\rho$ grows slower than $1/L$, all the dispersion terms vanish. 
\end{proof}
\begin{remark}[Global Uncertainty via Local Guarantees]
\label{rem:global-uncertainty}
    This theorem in conjunction with Lemma \ref{lem:mixtures-bic} has a powerful implication for agents with large uncertainty windows. At first glance, the drift requirement $\rho = o(1/L)$ may seem prohibitive. For example, consider an agent with \textit{global} uniform temporal belief $\mathcal{U}_{1,T}$ over the \textit{entire} horizon --- Theorem \ref{thm:uniform-first} would seem to suggest that we require the drift $\rho$ to be negligible relative to the horizon (roughly, $\rho \ll 1/T)$, effectively requiring a static environment for non-trivial incentive guarantees. However, we can represent $\mathcal{U}_{1,T}$ as a convex combination of uniform distributions over much shorter lengths $L$ (as long as $L = \tilde{\omega}(\sqrt{T})$)\footnote{Even if $T$ is not a multiple of the chosen $L$, we can still easily achieve an exact convex decomposition using blocks whose lengths differ from $L$ by at most a constant (e.g. blocks of length $L$ and $L+1$) and recover the same result since these blocks asymptotically incur error differing only by a constant factor.}. Applying the theorem above, we can achieve an IC-BNE for each of these component distributions, requiring only that $\rho \ll 1 / L \approx 1/\sqrt{T}$. Then, by Lemma \ref{lem:mixtures-bic}, this guarantee extends to the global belief $\mathcal{U}_{1,T}$ (and more generally, to any uniform belief over an interval of length at least $L$) without strengthening the drift requirement to $\rho = o(1/T)$. Equivalently, choosing $L = \tilde{\omega}(\sqrt{T})$ allows dynamic environments with drift as large as $\rho = \tilde{o}(T^{-1/2})$ while still obtaining $\epsilon = o(1)$ incentive compatibility for agents with global uncertainty, significantly relaxing the stationarity requirements. 
\end{remark}

This setting facilitates a rich class of belief distributions, allowing incoming agents to believe themselves uniform on \textit{any} contiguous subsequence of the history, achieving vanishing IC error even in dynamic environments. However, for bounds to be meaningful, the length of this subsequence must be sufficiently long --- roughly, larger than $\sqrt{T}$. If one is interested instead in agents whose temporal beliefs are on windows of a \textit{fixed} length $L$ which is a smaller growing function of $T$, we can tune Exp4.S to this target and leverage the sharper bounds from Theorem \ref{thm:luo} --- at the cost of not guaranteeing performance on substantially longer intervals:  
\begin{theorem}
    Define the set $U_L = \{\mathcal{U}_{s,L}: s \in [T - L + 1]\}$ as the set of uniform distributions over all contiguous intervals $I$ such that $|I| = L$. Fix a collection of agents with beliefs $\{(\mathcal{P}^i, \mathcal{D}^i)\}_{i=1}^T$ all of which satisfy Assumptions \ref{ass:adversarial1} and \ref{ass:slowmoving} with parameters $(\alpha, \Delta,\rho)$, and such that $\mathcal{D}^i \in U_L$ for all $i \in [T]$. Exp4.S with parameter $L$ implements an incentive-compatible Bayes $\epsilon(T)$-Nash equilibrium where the error term for agent $i$ is:
    \begin{equation*}
        \epsilon_i(T) = \frac{\left(\tilde{O}\left(L^{-1/2}\right) + \frac{2\rho(L^2 - 1)}{3L} \right)(\Delta - \rho(L-1))}{\alpha\left(\Delta - \tilde{O}(L^{-1/2}) - \rho(L-1)\right)}
    \end{equation*}
    When $\rho = o(1/L)$, we have vanishing error $\epsilon_i(T) = o(1)$. 
\end{theorem}
This instantiation of Exp4.S implements an IC-BNE for agents with uniform beliefs over windows of length $L$, for any (growing) choice of $L$. As in the previous setting, this incentive guarantee extends to mixtures over $U_L$, which are potentially more diffuse. However, the corresponding regret guarantees are local, and do not, in general, imply sublinear regret over the full horizon. 

These results give us a concrete picture of what incentive error looks like for a canonical family of temporal beliefs (uniform windows). More broadly, they illustrate a useful analysis principle: it suffices to establish incentive guarantees for a structured \textit{basis class} of belief distributions, in order to ensure the same guarantees for the richer class of mixtures. This suggests two natural questions: (i) Are there other, potentially more expressive, basis classes of distributions that yield good IC bounds, and (ii) Do agents whose temporal beliefs are \textit{close} (by some distance metric) to the convex hull also inherit similar IC properties? We do not fully answer these questions here --- good approximation families for temporal distributions remain an area of interest --- but we offer a preliminary result in Appendix \ref{app:approx-dist} on how IC guarantees extend to temporal distributions that are approximated by the convex hull of the basis set.

\section{Conclusion}
We study Bayesian incentivized exploration in the setting where agents have beliefs not only over rewards but also over arrival time. We formalize a connection between \textit{swap regret} and incentive compatibility; zero swap-regret gives exact incentive-compatibility to agents who believe that their arrival time is uniformly random, and more generally, $\mathcal{D}$-weighted swap regret translates to approximate incentive-compatibility for agents with temporal belief $\mathcal{D}$. When agents believe rewards are slowly-moving, \textit{external} regret is sufficient for a similar guarantee. These results give a prior-agnostic route to incentive-compatible exploration, where agents may hold a diverse collection of beliefs and the mechanism designer need not be aware of any of them. It also opens the door to using existing algorithms in the rich literature on online learning which are \textit{already} approximately incentive-compatible, directly \textit{via} their regret guarantees; we illustrate this by example. The fact that we support any arrival belief in the \emph{convex hull} of the set of distributions over which we offer weighted regret lets us give guarantees for substantially non-stationary environments: for example, by offering adaptive regret guarantees (weighted regret guarantees over all contiguous subsequences of length $\approx \sqrt{T}$) we automatically give incentive compatibility guarantees to agents who believe that their arrival time is uniform over even larger subsequences, while accomodating substantially non-trivial reward-movement parameters $\rho \approx 1/\sqrt{T}$.  

Our bounds have an explicit dependence on dispersion measures of the temporal distribution, which are not always transparent. One natural avenue for future work is to characterize these dispersion terms for broad families of temporal distributions, thereby yielding sharper and more interpretable guarantees. Another interesting direction could be to identify a tractable, expressive \textit{basis class} of distributions (beyond uniform distributions of length $L$) whose convex hull covers or approximates a very broad set of temporal beliefs, so that it suffices to establish weighted regret guarantees on this class. 

\bibliographystyle{plainnat}
\bibliography{refs}

\appendix
\section{Additional Proofs}
\label{sec:proofs}
Here we provide proof details for the lemmas from Section \ref{subsec:lemmas} whose proofs were deferred from the main body.

\lemmaone* 
\begin{proof}
For any fixed $s \in [T]$, we have $ |\mu_{t,a} - \mu_{s,a}| \le \rho\,|t-s|$. Taking expectation over 
$s\sim\mathcal{D}$ gives
\begin{equation*}
\mu_{t,a} \geq \mathbb{E}_{s \sim \mathcal{D}}[\mu_{s,a}] - \rho\Psi_\mathcal{D}(t)
\end{equation*}
and the same argument for $b$ gives
$ \mu_{t,b} \le \mathbb{E}_{s \sim \mathcal{D}}[\mu_{s,b}]+ \rho\,\Psi_{\mathcal{D}}(t)$.
Subtracting from the first inequality,
\begin{equation*}
    \mu_{t,a} - \mu_{t,b} \geq \mathbb{E}_{s \sim \mathcal{D}}[\mu_{s,a} - \mu_{s,b}] - 2\rho\Psi_\mathcal{D}(t) \geq \text{gap}_\mathcal{\mu, \mathcal{D}}(a) - 2\rho\Psi_\mathcal{D}(t).
\end{equation*}
\end{proof}

\lemmatwo*
\begin{proof}
    As in Lemma \ref{lem:pseudo-from-realized}, for all $t \in [T]$, define $\mathcal{F}_{t-1}$ as the $\sigma$-field generated by the history up to time $t-1$ (when running $\mathcal{A}$ on the random reward sequence generated according to $\mu$). Since $\tilde{I}_t$ is $\mathcal{F}_{t-1}$-measurable, we have:
    \begin{equation*}
        \mathbb{E}[\tilde{u}_{t, \phi(\tilde{I}_t)} - \tilde{u}_{t, \tilde{I}_t} | \mathcal{F}_{t-1}] = \mu_{t, \phi(\tilde{I}_t)} - \mu_{t, \tilde{I}_t}
    \end{equation*}
    for all $t \in [T]$. Taking the weighted sum (according to $\mathcal{D}$) and then taking expectation over algorithm randomness gives:
    \begin{equation*}
       \mathbb{E}\left[\sum_{t=1}^T \mathcal{D}(t) (\mu_{t, \phi(\tilde{I}_t)} - \mu_{t, \tilde{I}_t})\right]  = \mathbb{E}\left[\sum_{t=1}^T \mathcal{D}(t) (u_{t, \phi(\tilde{I}_t)} - u_{t, \tilde{I}_t})\right]   
    \end{equation*}
    The right-hand-side is upper-bounded by expected swap-regret, which gives the result. 
\end{proof}

\lemmathree*
\begin{proof}
Proceeding as in Lemma \ref{lem:prob-bound-max1}, we fix an action $a \in A$, and define the event $E_a := \{\mu: \text{gap}_{\mu, \mathcal{D}^i}(a) \geq \Delta \}$. For any reward instance $\mu \in E_a$, let $\mathcal{A}$ generate the hypothetical transcript $\tilde{\Pi}_T = \{(\tilde{I}_t, \tilde{a}_t, \tilde{u}_t)\}_{t=1}^T$ under rewards generated by $\mu$. By Lemma \ref{lem:ubound-general},
\begin{equation*}
    \mu_{t,a} - \mu_{t,b} \geq \Delta - 2\rho\Psi_{\mathcal{D}^i}(t)
\end{equation*}
for any $b \neq a$, and so:
\begin{equation*}
    \mu_{t,a} - \mu_{t,\tilde{I}_t} \geq (\Delta - 2\rho\Psi_{\mathcal{D}^i}(t))\mathbf{1}[\tilde{I}_t \neq a]
\end{equation*}
for all $t \in [T]$. Taking sums on both sides weighted by $\mathcal{D}^i$,
\begin{align*}
    \sum_{t=1}^T \mathcal{D}^i(t)\left(\mu_{t,a} - \mu_{t,\tilde{I}_t} \right) &\geq \Delta \sum_{t =1}^T \mathcal{D}^i(t) \mathbf{1}[\tilde{I}_t \neq a] - 2\rho \sum_{t =1}^T \mathcal{D}^i(t) \Psi_{\mathcal{D}^i}(t) \mathbf{1}[\tilde{I}_t \neq a] \\
    &\geq \Delta \sum_{t =1}^T \mathcal{D}^i(t) \mathbf{1}[\tilde{I}_t \neq a] - 2\rho\Phi(\mathcal{D}^i)
\end{align*}
because $\sum_{t=1}^T \mathcal{D}^i(t) \Psi_{\mathcal{D}^i}(t) \mathbf{1}[\tilde{I}_t \neq a] \leq \sum_{t=1}^T \mathcal{D}^i(t) \Psi_{\mathcal{D}^i}(t) = \Phi(\mathcal{D}^i)$.
Again, we bound the left-most term above by pseudo-regret:
\begin{equation*}
    \textrm{PReg}_{\mathcal{D}^i}(\tilde{\Pi}_T, \mu) \geq \Delta \sum_{t =1}^T \mathcal{D}^i(t) \mathbf{1}[\tilde{I}_t \neq a] - 2\rho \Phi(\mathcal{D}^i)
\end{equation*}
Taking expectation over algorithm randomness and using Lemma \ref{lem:pseudo-from-realized},
\begin{align*}
     \R(T) + 2\sqrt{2\,W_2(\mathcal D^i)\,\ln(2K)} &\geq \Delta \cdot \mathbb{E}_\mathcal{A}\left[\sum_{t=1}^T \mathcal{D}^i(t) \mathbf{1}[\tilde{I}_t \neq a]\right] - 2\rho\Phi(\mathcal{D}^i) \\
     &= \Delta(1 - \mathbb{E}_{\mathcal{A}}[N_a]) - 2\rho\Phi(\mathcal{D}^i)
\end{align*}
where $N_a = \sum_{t =1}^T \mathcal{D}^i(t) \mathbf{1}[\tilde{I}_t = a]$ is the probability under $\mathcal{D}^i$ of a recommended action being $a$. Rearranging, we get
\begin{equation*}
    \mathbb{E}_\mathcal{A}[N_a] \geq 1 - \frac{\R(T) + 2\sqrt{2\,W_2(\mathcal D^i)\,\ln(2K)} + 2\rho\Phi(\mathcal{D}^i)}{\Delta}
\end{equation*}
This bound is satisfied for every $\mu \in E_a$, and by Assumption \ref{ass:adversarial1}, $\Pr(E_a) \geq \alpha$. Since $N_a \geq 0$ by definition, taking an expectation also over $\mathcal{P}$, we have
\begin{align*}
    \mathbb{E}[N_a] \geq \Pr(E_a)\;\mathbb{E}[N_a |\; E_a] &\geq \alpha\left( 1 - \frac{\R(T) + 2\sqrt{2\,W_2(\mathcal D^i)\,\ln(2K)} + 2\rho\Phi(\mathcal{D}^i)}{\Delta}\right)
\end{align*}
\end{proof}

\section{Extensions for Approximated Temporal Beliefs}
\label{app:approx-dist}
In our main results we tied weighted regret guarantees to approximate incentive-compatibility, and the specific weightings are tied closely to our error bounds. Either we must know the agents' temporal distributions in advance, or we must assume they are in the convex hull of some basis distribution that we provide weighted-regret guarantees with respect to. Here, we note that since the weighted regrets of two ``close'' distributions (in terms of total variation) will also be close, it suffices for the agents' temporal distributions to be close to \textit{some} distribution in the convex hull of the basis distribution in order for us to say something about incentive-compatibility.

\begin{lemma}
\label{lem:tv-regret-transfer}
Let $\mathcal D_1,\mathcal D_2$ be two temporal beliefs. For any realized transcript $\Pi_T = \{(I_t, a_t, u_t)\}_{t=1}^T$ generated by some algorithm $\mathcal{A}$,
\begin{equation*}
|\mathrm{Reg}_{\mathcal D_2}(\Pi_T) -
\mathrm{Reg}_{\mathcal D_1}(\Pi_T)| \leq \;2\,d_{\mathrm{TV}}(\mathcal D_1,\mathcal D_2)
\end{equation*}
where $d_{\mathrm{TV}}(\mathcal D_1,\mathcal D_2)
:= \tfrac12\sum_{t\in[T]}|\mathcal D_1(t)-\mathcal D_2(t)|$ is the total variation distance between the two distributions. Consequently,
\begin{equation*}
\left|\mathbb{E}[\mathrm{Reg}_{\mathcal D_2}(\Pi_T)] -
\mathbb{E}[\mathrm{Reg}_{\mathcal D_1}(\Pi_T)]\right| \leq \;2\,d_{\mathrm{TV}}(\mathcal D_1,\mathcal D_2),
\end{equation*}
taking expectation over algorithmic randomness.
\end{lemma}
\begin{proof}
For any fixed sequence of values $\{I_t\}_{t=1}^T$, and any fixed action $a \in A$, we have:
\begin{align*}
    \sum_{t=1}^T \mathcal{D}_2(t) (u_{t,a} - u_{t,I_t}) &= \sum_{t=1}^T \mathcal{D}_1(t) (u_{t,a} - u_{t,I_t}) + \sum_{t=1}^T (\mathcal{D}_2(t) - \mathcal{D}_1(t)) (u_{t,a} - u_{t,I_t}) \\
    &\leq \sum_{t=1}^T \mathcal{D}_1(t) (u_{t,a} - u_{t,I_t}) + \sum_{t=1}^T |(\mathcal{D}_2(t) - \mathcal{D}_1(t))|\cdot|(u_{t,a} - u_{t,I_t})| \\
    &\leq \sum_{t=1}^T \mathcal{D}_1(t) (u_{t,a} - u_{t,I_t}) + \sum_{t=1}^T |(\mathcal{D}_2(t) - \mathcal{D}_1(t))|\\
    &= \sum_{t=1}^T \mathcal{D}_1(t) (u_{t,a} - u_{t,I_t}) + 2d_{\mathrm{TV}}(\mathcal D_1,\mathcal D_2)
\end{align*}
using the fact that $u_t \in [0,1]^K$ for all $t \in [T]$. Taking the max across all actions $a \in A$ on both sides,
\begin{equation*}
    \mathrm{Reg}_{\mathcal{D}_2}(\Pi_T) \leq \mathrm{Reg}_{\mathcal{D}_1}(\Pi_T) + 2d_{\mathrm{TV}}(\mathcal D_1,\mathcal D_2)
\end{equation*}
This is true in both directions, giving the desired result. Since this holds for each realized transcript, it holds after taking expectations as well.
\end{proof}

\begin{corollary}
\label{lem:BIC-approx}
Fix an algorithm $\mathcal{A}$, a class of distributions $\mathcal{C}$, and a collection of agents with reward and temporal beliefs $\{(\mathcal{P}^i, \mathcal{D}^i\}_{i=1}^T$. Suppose there exists $\beta \geq 0$ such that for each $\mathcal{D}^i$, there exists some distribution $C^i \in \text{conv}(\mathcal{C})$ in the convex hull of $\mathcal{C}$ such that $d_{\mathrm{TV}}(\mathcal{D}^i, C^i) \leq \beta$. If an algorithm $\mathcal{A}$ played over $T$ rounds guarantees expected $C$-weighted external regret $B(T)$ uniformly for all $C \in \mathcal{C}$, then $\mathcal{A}$ implements an incentive-compatible Bayes $\epsilon(T)$-Nash equilibrium with
\begin{equation*}
    \epsilon(T) = \min\{\eta_{\Psi}(T), \eta_{\Phi}(T), 1\}.
\end{equation*}
where we define
\begin{equation*}
    \eta_{\Psi}(T) = \frac{(B(T) + 2\beta +  2\rho\Phi(\mathcal{D}^i) + 2\sqrt{2W_2(\mathcal{D}^i)\ln(2K)})(\Delta - 2\rho\Psi_{\text{max}}(\mathcal{D}^i))}{\alpha(\Delta - 2\rho\Psi_{\text{max}}(\mathcal{D}^i) - B(T) - 2\beta -2\sqrt{2W_2(\mathcal{D}^i)\ln(2K)})}
\end{equation*}
\begin{equation*}
    \eta_{\Phi}(T) = \frac{\Delta(B(T) + 2\beta +  2\rho\Phi(\mathcal{D}^i) + 2\sqrt{2W_2(\mathcal{D}^i)\ln(2K)})}{\alpha(\Delta - 2\rho\Phi(\mathcal{D}^i) - B(T) - 2\beta -2\sqrt{2W_2(\mathcal{D}^i)\ln(2K)})}
\end{equation*}
\end{corollary}
\begin{proof}
    By convexity of expectation, if $\mathcal{A}$ achieves expected $C$-weighted external regret $B(T)$ uniformly for all $C \in \mathcal{C}$, it also achieves that $B(T)$ bound with respect to any distribution $C' \in \text{conv}(\mathcal{C})$. By Lemma \ref{lem:tv-regret-transfer},
    \begin{equation*}
        \R(T) \leq B(T) + 2d_{\mathrm{TV}}(\mathcal{D}^i, C^i) \leq B(T) + 2 \beta
    \end{equation*}
    Substituting this into the bound of Theorem \ref{thm:BIC-ext} gives the result. 
\end{proof}

Corollary \ref{lem:BIC-approx} is most informative when a large class of agent beliefs are close in TV to the convex hull of a smaller basis class we can analyze for incentive-compatibility directly. Designing such basis classes with better approximation properties (which retain efficient weighted regret guarantees) is a compelling direction for future work. 
\end{document}